\documentclass{article}
\date{}
\usepackage{amsfonts}
\usepackage{amsmath}
\usepackage{graphicx, hyperref}

\usepackage[a4paper]{geometry}
\usepackage{graphicx}
\usepackage{amsmath}
\usepackage{amssymb}
\usepackage{amsfonts}
\usepackage{amsopn,ntheorem}
\usepackage{tikz}

\usetikzlibrary{fit}					
\usetikzlibrary{backgrounds}	
\newcommand{\tr}{\text{\rm{tr}}}

\newtheorem{theorem}{Theorem}
\newtheorem{lemma}{Lemma}

\newtheorem{definition}{Definition}

\newtheorem{remark}{Remark}
\newtheorem{problem}{Problem}

\newenvironment{proof}{{\noindent{\bf Proof:}}}{$\hfill\Box$}

\def\n{\nonumber\\}
\def\be{\begin{equation}}
\def\ee{\end{equation}}
\def\bes{\begin{equation*}}
\def\ees{\end{equation*}}
\def\ma{{\mathcal A}}
\def\mb{{\mathcal B}}

\def\mR{{\mathcal R}}
\def\mS{{\mathcal S}}

\def\mx{{\mathcal X}}
\def\my{{\mathcal Y}}

\def\bF{\mathbf{F}}
\def\bE{\mathbf{E}}
\def\bJ{\mathbf{J}}

\def\bQ{\mathbf{Q}}

\def\tv#1{\left\|#1\right\|_1}

\newcommand{\ket}[1]{|#1\rangle}
\newcommand{\bra}[1]{\langle#1|}

\def\ConvHull{\text{\rm{ConvHull}}}

\newcommand{\Bob}{\text{Bob}}

\begin{document}
\title{On Dimension Bounds for Auxiliary Quantum Systems}

\author{Salman Beigi\\ {\it \small School of Mathematics,} {\it \small Institute for Research in Fundamental Sciences (IPM),} {\it \small Tehran, Iran}
\and Amin Gohari \\ {\it \small Department of Electrical Engineering,} {\it \small Sharif University of Technology,} {\it \small Tehran, Iran}\\
{\it \small School of Computer Science,} {\it \small Institute for Research in Fundamental Sciences (IPM),} {\it \small Tehran, Iran}}

\maketitle

\begin{abstract}
Expressions of several capacity regions in quantum information theory involve an optimization over auxiliary quantum registers. Evaluating such expressions requires bounds on the dimension of the Hilbert space of these auxiliary registers, for which no non-trivial technique is known; we lack  a quantum analog of the Carath\'{e}odory theorem. In this paper, we develop a new non-Carath\'{e}odory-type tool for evaluating expressions involving a single quantum auxiliary register and several classical random variables. As we show, such expressions appear in problems of entanglement-assisted Gray-Wyner and entanglement-assisted channel simulation, where the question of whether entanglement helps in these settings is related to that of evaluating expressions with a single quantum auxiliary register. To evaluate such expressions, we argue that developing a quantum analog of the Carath\'{e}odory theorem requires a better understanding of a notion which we call ``quantum conditioning." We then proceed by proving a few results about quantum conditioning, one of which is that quantum conditioning is strictly richer than the usual classical conditioning.
\end{abstract}

\section{Introduction}
One of the central goals of information theory is to find computable expressions for capacity regions of problems involving the transfer of information. An expression is computable if for every $\epsilon>0$, there is an algorithm that stops in finite time $T_{\epsilon}$ and outputs an approximation of the expression within $\epsilon$. All computable capacity regions that have been found so far turn out to be expressible in a so-called \emph{single-letter form} where a union is taken over a finite set of auxiliary random variables or auxiliary quantum registers.\footnote{Here given a (classical or quantum) density matrix $\rho_{\mathbf{A}_1\dots \mathbf A_k}$ by an auxiliary register we mean and extra subsystem $\mathbf E$ and an extension density matrix $\rho_{\mathbf{A}_1\dots \mathbf A_k \mathbf E}$ whose marginal on $\mathbf{A}_1\dots \mathbf A_k$ is the starting density matrix.}
 So for computability we need a restriction on the dimension or cardinality of these auxiliary registers, but no general tool is known for proving such bounds especially in the quantum case.

   The problem of bounding the dimension of auxiliary quantum registers has arisen in the literature. A single-letter (additive) formula for the entangling capacity of a bipartite unitary is provided in \cite{Bennett93}. This formula involves an optimization over unbounded auxiliary quantum registers, and the question of how large the auxiliary systems need to be in the optimal protocol is still open~\cite{Bennett93}. The problems of squashed entanglement measure \cite{Christandl} and quantum channel capacity assisted with symmetric side channels \cite{Smithetal} provide two other such examples.

To be more precise let us explain the problem of squashed entanglement in more details. Squashed entanglement of a bipartite state $\rho_{\mathbf{AB}}$ is defined by
$$E_{sq}(\rho_{\mathbf{AB}}) =\frac 12\inf_{\rho_{\mathbf{ABF}}}I(\mathbf{A};\mathbf{B}|\mathbf{F}).$$
where the infimum is taken over all extensions $\rho_{\mathbf{ABF}}$ of $\rho_{\mathbf{AB}}$.
To compute squashed entanglement through brute-force search, a dimension bound on auxiliary register $\mathbf{F}$. This is indeed the reason that even proving the faithfulness of squashed entanglement is a hard problem~\cite{BCY}.

Let us consider another example involving one auxiliary quantum register. Given two random variables $X$ and $Y$, consider the region formed by pairs $\big(H(X|\textbf{F}), H(Y|\textbf{F})\big)$ when we take the union over all auxiliary registers $\textbf{F}$.
When $\textbf{F}$ is a quantum register, evaluating this region by numerical brute-force simulation requires a bound on the dimension of $\textbf{F}$. Given such a bound we may compare this region with the classical case where $\textbf{F}$ is taken to be an auxiliary random variable. In particular it is interesting to verify whether all points of this region can be obtained from classical registers $\textbf{F}$ or not.

An arbitrary expression (even in the classical world) may not admit a dimension bound on its auxiliary variables. To prove such bounds, classical information theory provides us with a few tools, mainly the Carath\'{e}odory theorem, but also the perturbation method \cite[Appendix C]{AbbasYoungHan}, \cite{Gohari} and some manipulation techniques as in \cite{Nair}. Not much is known when these techniques fail.  The situation is more murky in the quantum world. It is fair to say that no non-trivial technique for bounding the dimension of auxiliary quantum registers is known, even for the above simple example.

\subsection{Shared entanglement in classical scenarios}
Our main motivation for studying the problem of dimension bounds is to understand the benefit of entanglement is classical communication settings. It is well-known that entanglement does not help in the problem of transmission of classical information over a point-to-point classical channel~\cite{E-assisted}, but is there an information theoretic classical setting in which shared entanglement helps? To avoid confusions let us explain what we mean by an information theoretic classical scenario. By a classical scenario we mean settings with classical inputs, outputs, and communication channels, implying that all variables are classical random variables except perhaps shared entanglement. Moreover, by an information theoretic setting we mean a framework involving average quantities over repeated trials permitting an \emph{asymptotically} vanishing error. From this perspective, for instance, the increase of the \emph{zero-error} capacity of classical channels by entanglement~\cite{Zero1, Zero2} does not fit into our framework. Similarly violation of Bell's inequalities~\cite{Bell} in the presence of entanglement is not an example of our scenarios. Here we explain two examples in which understanding the benefit of shared entanglement is of interest.

Our first example is the Gray-Wyner problem whose goal is to transmit multiple correlated sources to multiple distant
parties~\cite{GrayWyner}. Recently Winter (personal communication, 2012) has found the rate region of the
entanglement-assisted Gray-Wyner problem, yet we do not know whether entanglement helps in this scenario or not. The region found by Winter involves taking a union over several auxiliary quantum systems. If we assume that these auxiliary registers are all classical random variables, we obtain the classical rate region of the Gray-Wyner problem~\cite{GrayWyner}. In Section~\ref{Appendix:GW}, we introduce a related region with a \emph{single} auxiliary quantum system. We will see that if by replacing this single auxiliary quantum register by a classical one we obtain the same region, then entanglement does not help in the Gray-Wyner problem. Indeed if we could solve optimization problems over an auxiliary quantum register then we could compare the answer to the classical case and decide whether shared entanglement helps in the Gray-Wyner problem or not.

Our second example is the problem of simulating bipartite correlations via communication in the presence of shared entanglement. That is, how much communication is required to simulate a given bipartite correlation when the two parties are provided with shared entanglement. While the overall task of generating correlations is classical in nature, shared entanglement is known to help in a non-information theoretic setup; for the case of no communication, Bell's theorem states that there are bipartite correlations that can be generated in the presence of entanglement. But how about Bell's scenarios in an information theoretic setting in which we repeat the Bell experiment many times and allow for  asymptotically vanishing error? One of our main results in this paper is to show that entanglement can still help. By extending the result of \cite{Yassaee} (see Theorem~\ref{thm:q-region} below) we show that the amount of communication required to simulate a bipartite correlation in the presence of entanglement can be expressed as an optimization problem over a quantum auxiliary register. This characterization again fits into the framework of expressions for which we are interested in a dimension bound. Although we do not have a dimension bound here, by an ad hoc indirect approach we show that entanglement does help in Bell's scenarios even in an information theoretic setting (see Appendix~\ref{sec:CHSHExample}).

\subsection{A new tool}\label{sec:introtool}
The few known tools from classical information theory for bounding the dimension of auxiliary registers are not readily applicable to the quantum setting. The reason is that the main tool from the classical theory, namely the Carath\'{e}odory theorem, heavily relies on the fact that for any two random variables $X$ and $C$, the conditional entropy $H(X|C)$ can be written as the convex linear expression $\sum_{c}p(c)H(X|C=c)$; i.e., classical conditioning is a simple convexification. But re-expressing $H(X|\mathbf{F})$ as such a convex combination is invalid when we condition on a quantum register $\textbf{F}$. We believe that any proper dimension bound in the quantum case has to provide insights into ``quantum conditioning," i.e., conditioning on a quantum register. Thus, understanding quantum conditioning is central to any potential use of Carath\'{e}odory theorem in bounding dimensions. 

To understand quantum conditioning we develop a new tool to study optimization problems involving quantum registers. This is the first bounding tool in quantum information theory, and is based on a reformulation of the problem in terms of a max-min expression, followed by using a minimax theorem (see Section \ref{sec:DB}  or the proof of Theorem \ref{thm:QCMain1} for the use of the tool). This tool is inspired by a work on classical broadcast channels \cite[Sec. III.B]{GCA2012}.
Below we briefly explain some implications of this tool.\\

\noindent
\textbf{An optimization problem and the mutual information curve:}
Consider the following two problems:
\begin{problem}
Given an arbitrary distribution $q(x,y,z)$, consider the optimization
\begin{align*}
\sup_{C-X-YZ}I(C;Y)-I(C;Z),
\end{align*}
over all classical random variables $C$; thus we basically take the supremum over all channels $p(c|x)$. This expression shows up in several network information theoretic problems, especially those involving security (e.g., Wiretap channel problem \cite[Theorem 22.1]{AbbasYoungHan}). Here we ask a standard cardinality reduction question: can we find the smallest cardinality bound $d^*$ on the alphabet size of $C$ that universally works for all $q(x,y,z)$?
\end{problem}

\begin{problem}
Suppose we are given sets $\mathcal{X}$ and $\mathcal{C}$. Then given a channel $p(c|x)$, we can consider the function $p(x)\mapsto I(X;C)$, i.e., the curve (or surface) of mutual information versus the input distribution. We know that this curve is concave and its maximum is equal to the channel capacity of $p(c|x)$. Now let us consider all the mutual information curves produced by all channels $p(c|x)$ with output cardinality of $C$ bounded from above by $d$, i.e., with $|\mathcal C|\leq d$. We use $\Gamma(d)$ to denote this set of all $p(x)\mapsto I(X;C)$ curves for channels with output alphabet of size less than or equal to $d$. Let $\ConvHull(\Gamma(d))$ be the convex hull of $\Gamma(d)$, i.e., the set of all curves of the form $p(x)\mapsto \sum_{i}\omega_i I(X;C_i)$ for arbitrary non-negative weights $\omega_i$ that add up to one, and $p(c_i|x)$ with output cardinality bounded from above by $d$. Clearly $$\ConvHull(\Gamma(d))\subseteq \ConvHull(\Gamma(d+1)).$$ If $d=1$, the output of the channel $C$ has to be a constant random variable and $p(x)\mapsto I(X;C)=0$ is a trivial line (hyperplane). Thus $\Gamma(1)=\ConvHull(\Gamma(1))$ contains only a single curve. As we increase $d$, we allow for more complicated channels and as such the behavior of the curve $p(x)\mapsto I(X;C)$ can become more complicated. Let $d^*$ be the smallest integer such that $\ConvHull(\Gamma(d^*))= \ConvHull(\Gamma(d))$ for all $d\geq d^*$. If no such $d^*$ exists, we set it to be infinity.
\end{problem}

Our new tool implies that the answer for $d^*$ to the above two problems is the same (see Section~\ref{sec:DB}). Thus, we can compute $d^*$ for the second question from $d^*$ for the first question. By standard techniques \cite[Appendix C]{AbbasYoungHan}, the cardinality bound of $d^*\leq |\mathcal{X}|$ on the alphabet set of $C$ can be imposed in the first problem. Thus the same $d^*\leq |\mathcal{X}|$ works also for the second example.

Similarly, we can state these two problems in the quantum case:
\begin{problem}
Given an arbitrary $q(x,y,z)$ compute
\begin{align*}
\sup_{\textbf{F}-X-YZ}I(\textbf{F};Y)-I(\textbf{F};Z),
\end{align*}
over all quantum registers $\textbf{F}$, where $\textbf{F}-X-YZ$ represents the Markov chain condition $I(\bF; XZ| X)=0$. 
Can we bound the dimension of $\bF$ in this optimization problem, and if yes, can we find a minimum dimension bound $d^*$ on $\mathbf{F}$ that universally works for all $q(x,y,z)$?
\end{problem}
\begin{problem}
Suppose that we are given a fixed set $\mathcal{X}$. Given a classical-quantum (c-q) channel $X\mapsto \mathbf{F}$, we can consider the curve $p(x)\mapsto I(X;\mathbf{F})$, i.e., the curve of mutual information versus the input distribution. Again we know that this curve is concave and its maximum is equal to the channel capacity. Now let us consider all the mutual information curves produced by all c-q channels  $X\mapsto \mathbf{F}$ with output dimension of $\mathbf{F}$ bounded from above by $d$, and denote it by $\Gamma^q(d)$. We write the convex hull of $\Gamma^q(d)$ as $\ConvHull(\Gamma^q(d))$. Clearly we have 
$$\ConvHull(\Gamma^q(d))\subseteq \ConvHull(\Gamma^q(d+1)).$$
 Let $d^*$ be the smallest integer such that $\ConvHull(\Gamma^q(d))= \ConvHull(\Gamma^q(d^*))$ for all $d\geq d^*$. If no such $d^*$ exists, we set it to be infinity.
\end{problem}
Again the answers to these questions are the same (see Section~\ref{sec:DB}). Thus, to study Problem 3 we can look at Problem 4.

By the same tool we show in Theorem \ref{thm:QCMain1} that there exists a distribution $q(x,y,z)$ such that the supremum over auxiliary quantum registers $\textbf{F}$ in the above optimization problem yields a larger value than taking the maximum of the same expression over classical auxiliary random variables. In other words, there exists a distribution $q(x,y,z)$ such that
\begin{align}\label{eq:ineq-xyz}
\sup_{\textbf{F}-X-YZ \atop
\textbf{F} \text{ quantum}}I(\textbf{F};Y)-I(\textbf{F};Z) >  \max_{C-X-YZ \atop C \text{ classic}}I(C;Y)-I(C;Z).
\end{align}
To find such $q(x,y,z)$ we construct a channel $X\mapsto \textbf{F}$ whose $p(x)\mapsto I(X;\mathbf{F})$ curve does not belong to $\ConvHull(\Gamma(d^*))$, the convex hull of curves of the form $p(x)\mapsto I(X;C)$ for classical $p(c|x)$ channels. Interestingly, our example is based on Kochen-Specker sets, and the fact that shared entanglement increases the one-shot zero-error capacity of classical channels~\cite{Zero1}.

In the above example the size of $X$ is large (as large as the size of the smallest Kochen-Specker set).
We show that when $X$ is binary and the dimension of $\textbf{F}$ is two, then any mutual information curve $p(x)\mapsto I(X;\mathbf{F})$ is equal to the $p(x)\mapsto I(X;C)$ curve of a carefully constructed classical channel $p(c|x)$ (see Theorem \ref{thm:QCMain2} part 1). Here, we have a non-trivial mapping from an arbitrary register $\textbf{F}$ of size two, to a classical random variable $C$. We use this to prove (see Theorem \ref{thm:QCMain2} part 2) that for any channel $q(y,z|x)$ we have
\begin{align*}
\sup_{\textbf{F}-X-YZ\atop \dim \bF=2}I(\textbf{F};Y)-I(\textbf{F};Z) = \max_{C-X-YZ}I(C;Y)-I(C;Z).
\end{align*}

\vspace{.1in}
\noindent
\textbf{Quantum conditioning is richer than classical conditioning:}
Our example of a distribution $q(x,y,z)$ with \eqref{eq:ineq-xyz} indeed shows that quantum conditioning as defined above, is strictly richer than classical conditioning (Theorem \ref{thm:QCMain1} part (b)). More precisely there exist random variables $X_1, \dots, X_m$ and  quantum register $\mathbf{F}$ such that for every auxiliary random variable $C$ we have
\begin{align*}
\left(H(X_1| C), \dots, H(X_m| C)\right)
\neq \left(H(X_1| \mathbf{F}), \dots, H(X_m| \mathbf{F})\right).
\end{align*}
\subsection{Organization of the paper}

The reminder of this paper is organized as follows. In Section~\ref{sec:prelim} we set up our notation and remind some preliminaries. In Section \ref{sec:ChSim} we provide two examples, motivating looking at expressions with a single auxiliary quantum register. The reader may choose to skip this section and continue with Section~\ref{sec:QC} that contains some of the main results on dimension bounds. This section sets up a framework for discussing quantum conditioning (or ``quantum convexification"). It is followed by an example that quantum conditioning is strictly richer than classical conditioning. Section~\ref{sec:DB} discusses the potential use of our technique to bound the dimension of auxiliary quantum registers. Some technical details come in the Appendix.

\section{Notations}\label{sec:prelim}
Classical random variables are denoted by capital letters $A, B, X, Y$. The set of outcomes (alphabets) of $X$ is denoted by $\mathcal{X}$, and by size of $X$ we mean $|\mathcal{X}|$, size of the set $\mathcal{X}$. By $X^n=X_1\dots X_n$ we mean $n$ i.i.d.\ copies of $X$, and $X_{\ell:k}$ (for $k\geq \ell$) means $X_{\ell}X_{\ell+1} \dots X_k$.
Outcomes of $X^n$ are denoted by $x^n=x_1\dots x_n$, so the outcome of the $i$-th random variable in $X^n$ is $x_i\in \mathcal{X}$. The sequence $x^n=x_1\dots x_n$ happens with probability $p(x^n)=p(x_1)\cdots p(x_n)$ where $p(x)$ is the distribution of $X$.

To distinguish quantum registers from classical random variables we denote them by boldface letters $\bE, \bF$, and the dimension of the corresponding Hilbert space to $\bF$ is denoted by $\dim \bF$. Again $\bF^n = \bF_1\dots \bF_n$ denotes $n$ independent copies of $\bF$.

$H(\cdot)$ denotes the entropy function (either Shannon or von Neumann entropy), and $I(\cdot\, ; \cdot)$ is the mutual information
$$I(X; Y)= H(X) + H(Y)-H(XY).$$
Moreover $H(X|Y)= H(XY)- H(Y)$ is the conditional entropy.
For random variables $X, Y, Z$ by $X - Y - Z$ (the Markov chain condition) we mean that $I(X ; Z| Y) =0$, where $I(X; Z| Y)= H(X|Y) + H(Z|Y)- H(XZ|Y)$.
 We use the same notation if either of $X, Y$ or $Z$ is a quantum register.

 When $Y$ is classical $X - Y - Z$ equivalently means that $X$ can be generated out of $Y$ using a channel independent of $Z$. When $Y$ is quantum however, by applying a measurement on $Y$ to generate $X$ we destroy $Y$. So $X, Y$ do not simultaneously exist and in this case $I(X ; Z| Y)$ has no meaning. As a result we save the notation $X - Y - Z$ when all $X, Y, Z$ simultaneously exist and $I(X ; Z| Y)=0$. Lemma~\ref{lem:purification} in Appendix~\ref{app:useful-lem} gives an operational meaning to $X-\bF - Y$ when $\bF$ is a quantum register.

For either probability distributions or quantum states the norm-$1$ distance is denoted by $\|\cdot \|_1$.

Further notations are developed in Appendix~\ref{app:useful-lem}. These notations are not required to understand the body of the paper.

\section{Entanglement-assisted rate regions}\label{sec:ChSim}
In this section we provide two examples to motivate the study of expressions with a single quantum auxiliary system, and several classical random variables. The reader may choose to skip this section and continue to Section \ref{sec:QC}.

\subsection{Entanglement-assisted Gray-Wyner problem}
\label{Appendix:GW}
The Gray-Wyner problem involves $m+1$ parties, Alice and $m$ Bobs who we call $\Bob_1, \dots, \Bob_m$.
In this problem Alice observes i.i.d.\ copies of $X_1, \dots,X_m$, and her goal is to send $X_i$ to $\Bob_i$. That is,
$\Bob_i$ wants to recover the i.i.d.\ copies of $X_i$ with probability of error converging to zero as the number of i.i.d.\ observations goes to infinity (see Figure~\ref{fig:gray-wyner}). To do so, Alice can send a public message at rate $R_0$ to all Bobs, and $m$ private messages at rates $R_1, R_2, \dots, R_m$ (at rate $R_i$ to $\Bob_i$).

The Gray-Wyner region is defined to be the set of \emph{achievable} rate vectors $(R_0, R_1, \dots, R_m)$, where by achievable we mean that by sending public and private information at rates $R_0$ and $R_1, \dots, R_m$ respectively, Bobs' demands can be fulfilled.
The Gray-Wyner region, $\mathcal{R}^c_{\text{GW}}$, is characterized by the set of all tuples
\begin{align}\label{eq:gw-tuple}
\big(I(X_1\dots X_m;C), H(X_1|C), \dots, H(X_m|C)\big),
\end{align}
where  $C$ is an arbitrary auxiliary random variable~\cite{GrayWyner}. That is, $(R_0, R_1, \dots, R_m)$ is achievable if and only if there exists an auxiliary random variable $C$ such that $R_0\geq I(X_1\dots X_m;C)$ and $R_i\geq H(X_i|C)$.

\begin{figure*}
\begin{center}
\includegraphics[width=4in]{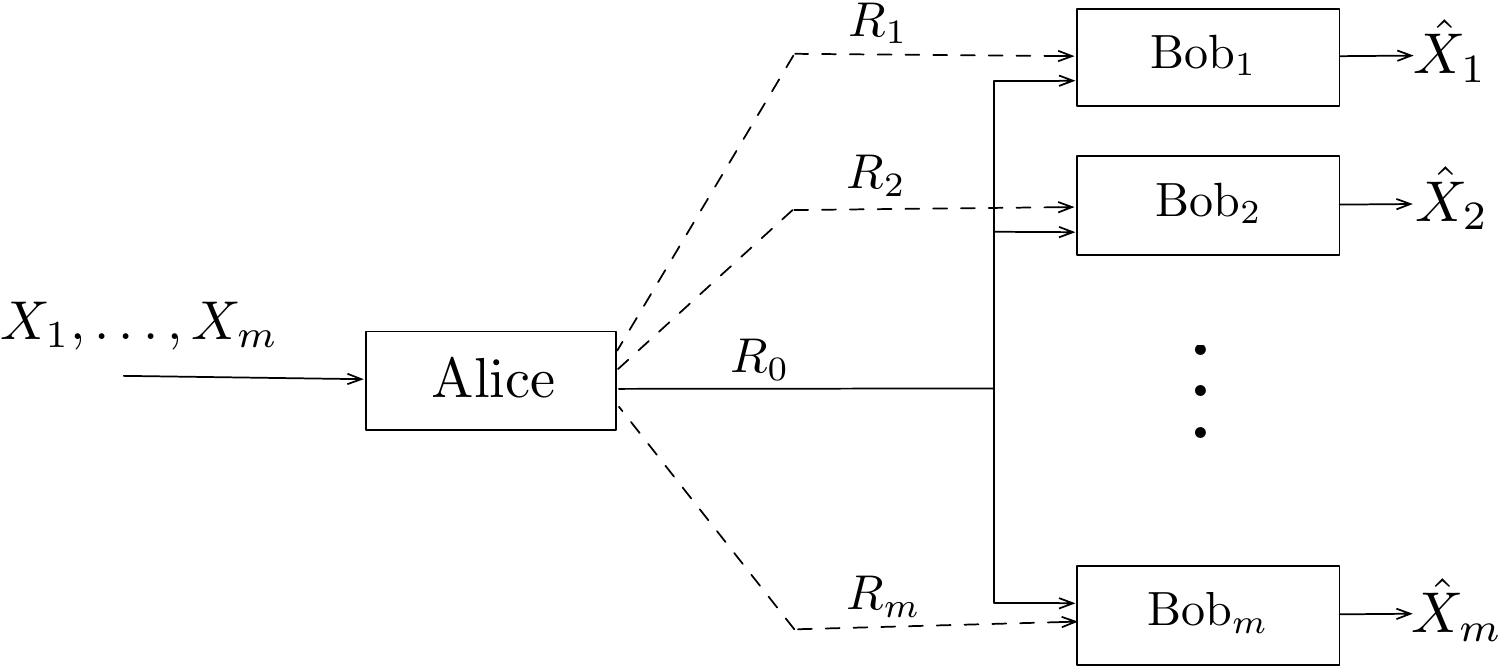}
\caption{The Gray-Wyner game consists of $m+1$ players, Alice and $m$ Bobs who are indexed by $i=1, \dots, m$. Alice receives the i.i.d.\ copies of
$ X_1, \dots, X_m,$ sends public information at rate $R_0$ to all Bobs and private information at rate $R_i$ to $\Bob_i$. The goal of $\Bob_i$ is to recover $X_i$.
}\label{fig:gray-wyner}
\end{center}
\end{figure*}

Recently Winter (personal communication, 2012) has shown that the rate region of the \emph{entanglement-assisted} Gray-Wyner problem, $\mathcal{R}^q_{\text{GW}}$, is characterized by all tuples
\begin{align}\label{eq:gw-tuple-r}
\big(I(X_1\dots X_m;\bF_1\dots \bF_m),  H(X_1|\bF_1), \dots, H(X_m|\bF_m)\big),
\end{align}
where $\bF_1,\dots , \bF_m$ are $m$ arbitrary auxiliary quantum registers.\footnote{Roughly speaking, the converse follows from similar steps as in the classical Gray-Wyner problem. The achievability follows from a remote state preparation protocol together with the quantum-classical Slepian-Wolf theorem~\cite{cqSW}. $\bF_i$ is $\Bob_i$'s part of the shared (multipartite) entangled state.} Here by entanglement-assisted we mean that Alice and Bobs can share any \emph{multipartite} entangled state as a resource.

Although we have a characterization of the entanglement-assisted Gray-Wyner rate region in a single letter form, we do not know whether this rate region is strictly larger that the classical rate region or not, i.e., whether entanglement helps in the Gray-Wyner problem or not. Indeed the inclusion $\mathcal R^c_{\text{GW}} \subseteq \mathcal R^q_{\text{GW}}$ is immediate, but we do not know whether it is strict or not. Our numerical simulations do not give any point in $\mathcal R^q_{\text{GW}}$ outside of $\mathcal R^c_{\text{GW}}$, but this does not give their equality. This is because our brute-force searches are in bounded dimensions and unlike the classical case,  we do not have a bound on the dimensions of $\bF_i$'s in~\eqref{eq:gw-tuple-r}.

Let us introduce a third region $\mathcal{R}^{q'}_{\text{GW}}$ characterized by tuples
\begin{align}\label{eq:gw-tuple}
\big(I(X_1\dots X_m;\bF), H(X_1|\bF), \dots, H(X_m|\bF)\big),
\end{align}
where $\bF$ is an arbitrary auxiliary quantum register. Note that $\mathcal{R}^{q'}_{\text{GW}}$ involves a single auxiliary quantum register.  Observe that $\mathcal{R}^c_{\text{GW}}\subseteq\mathcal{R}^{q}_{\text{GW}}\subseteq\mathcal{R}^{q'}_{\text{GW}}$ since by identify $\bF$ by $\bF_1\dots \bF_m$ we have $H(X_i|\bF_1)\geq H(X_i|\bF_1\dots \bF_m)$. Therefore, if we show that $\mathcal{R}^c_{\text{GW}} =\mathcal{R}^{q'}_{\text{GW}}$ then we conclude that entanglement does not help in the Gray-Wyner problem. Note that the problem of comparing $\mathcal{R}^c_{\text{GW}}$ and $\mathcal{R}^{q'}_{\text{GW}}$ is a special case of the following general problem: in an optimization problem involving one auxiliary quantum register and several classical random variables, do we get the same optimal value if we restrict the auxiliary register to be classical?

\subsection{Simulation of bipartite correlations}
Our second example (for motivating expressions with a single auxiliary quantum register) is the problem of simulating bipartite correlations with one-way classical communication.
In this problem Alice and Bob observe i.i.d.\ repetitions of two random variables $X$ and $Y$ respectively, and would like to generate i.i.d.\ repetitions of random variables $A$ and $B$ respectively. Here $X, Y, A, B$ are jointly distributed according to a given $p(a,b,x,y)=p(x, y)p(a, b\vert x, y)$.
For this simulation
there is a one-way communication link from Alice to Bob at rate $R$. The question is for which values of $R$ this simulation is feasible. In the following we describe this problem in details.\\

\noindent
\textbf{Classical case (with shared randomness):}
Assume that Alice and Bob observe i.i.d.\ copies of $X^n$ and $Y^n$ respectively jointly distributed according to $\prod_{i=1}^np(x_{i},y_{i})$. Here $p(x,y)$ is the joint distribution of $XY$. They also share common randomness $s$ at some arbitrary rate. An $(n, \epsilon,R)$ code consists of a randomized encoder $\tilde p(m|x^ns)$ and two randomized decoders $\tilde{p}(a^n|mx^ns)$, $\tilde{p}(b^n|my^ns)$ such that $\frac{1}{n}H(M)\le R$ and
\begin{align}\label{eq:e-c}
\tv{\tilde{p}(a^n, b^n, x^n, y^n) - \prod_{i=1}^n p(a_i,b_i,x_i,y_i)}\leq \epsilon,
\end{align}
where
\begin{align*}
\tilde p(a^n, b^n, x^n, y^n) =
 p(x^n, y^n)\sum_{s, m} p(s) \tilde p(m|x^n s) \tilde p(a^n | m x^n s) \tilde p(b^n|my^n s).
\end{align*}
A rate $R$ is said to be achievable if there exists a sequence of $(n, \epsilon_n, R)$ codes such that
\[
\lim_{n\rightarrow \infty}\epsilon_n =0.
\]
The set of achievable rates is denoted by $\mathcal{R}^c$.

Yassaee et al. \cite{Yassaee} solve a generalization of this problem. Their rate region (Theorem 1 of \cite{Yassaee}) reduces to the following region as a special case.

\begin{theorem}\label{thm:yassaee}
\begin{align}
\mathcal{R}^c =\bigg\{R: \ \exists p&(c| a, b, x, y):  \nonumber
\\&~R\geq I(X; C| Y),\n
&~C-X-Y,\n
&~A-CX-YB,\n
&~B-CY-XA, \n
&~|\mathcal C|\le|\mx||\my||\ma||\mb|+1\bigg\}. \label{eq:c-c}
\end{align}
Here $C$ is some auxiliary random variable with joint distribution $p(c, a, b, x, y) = p(a, b, x, y) p(c| a, b, x, y)$ satisfying the above constraints. In other words, $R$ is achievable if and only if $R\geq \min I(X; C| Y)$ where the minimum is taken over auxiliary random variables $C$ satisfying the conditions given in equation ~\eqref{eq:c-c}.
\end{theorem}

We should point out here that the problem of simulating bipartite correlations with classical communication has been appeared in the literature~\cite{Brassardetal99, Pironio, TonerBacon, Degorreetal05, Brassardetal06, Shi, RegevToner, Roland09, Gisin, Degorreetal11}. For instance Toner and Bacon~\cite{TonerBacon} show that correlations obtained from projective measurements on a singlet can be simulated by one bit of (classical) communication. Regev and Toner~\cite{RegevToner} prove that a generalization of these correlations can be simulated by two bits of communication. Moreover, Pironio~\cite{Pironio} computes the \emph{average} amount of communication required to simulated CHSH-type correlations. Also, Roland and Szegedy~\cite{Roland09} find the \emph{asymptotic} rate of one-way communication needed to simulate CHSH-type correlations. We will comment more on these results later in Appendix~\ref{sec:CHSHExample}, but here we would like to emphasize that none of these papers considers the problem in an information theoretic setting. Namely, some of these papers consider the problem of simulation in a single shot, and not in a parallel repetition form, and those that study the asymptotic case of the problem do not allow for asymptotically vanishing probability of error. Given the fact that the information theoretic setting is a relaxation of the previous settings, Theorem~\ref{thm:yassaee} provides a \emph{lower bound} on all of the above bounds.\\

\noindent
\textbf{Quantum case (with shared entanglement):}
The setup here is similar to the classical case except that instead of shared randomness, Alice and Bob are provided with
shared entanglement with an arbitrary rate. To simulate the correlation, Alice applies a measurement chosen according to the observed $x^n$, on her part of the shared entanglement. The measurement outcome has two parts: the first part is taken as $a^n$ and the second part is taken as the message $m$ to be transmitted to Bob. Bob uses $m$ and his observation $y^n$ to choose a measurement to be applied on his quantum system. The outcome of this measurement is taken as $b^n$.

A rate $R$ is achievable if there exists a sequence of codes $(n, \epsilon_n,R)$ as above such that $\frac{1}{n}H(M)\leq R$ and the total variation distance between the induced distribution $\tilde{p}(a^n,b^n,x^n,y^n)$ by the code and the original distribution $\prod_{i=1}^n p(a_i,b_i,x_i,y_i)$ is at most $\epsilon_n$, i.e.,~\eqref{eq:e-c} holds. The set of achievable rates in this case is denoted by $\mathcal{R}^q$.

For $\epsilon\geq 0$ define
\begin{align}\label{eq:f-q}
&\mathcal{S}_{\epsilon} = \bigg\{R: \ \exists \bF \text{ satisfying the following conditions: }\nonumber\\&
~~~~ \qquad R\geq I(X; \bF| Y),\nonumber\\
& ~~~~\qquad\widetilde \rho_{AXY\bF} = \sum_{a, x, y} \tilde p(a, x, y) \ket{a, x, y}\bra{a, x, y}\otimes \rho_{a, x,y}^{\bF},\nonumber \\
& ~~~~\qquad\tilde{p}(a, b, x, y)= \tilde{p}(a, x, y)\tilde{p}(b| a, x, y),\nonumber \\
&~~~~\qquad\tv{\tilde{p}(a,b,x,y)-p(a,b,x,y)}\leq \epsilon,\n
&~~~~\qquad\textbf{F}-X-Y, ~A-\textbf{F}X-Y,\n
&~~~~\qquad\exists \Psi \text{ s.t. } \Psi(\textbf{F},Y)=(B,Y)\bigg\}.
\end{align}
By the second and third constraints we mean that there is a joint distribution $\tilde{p}(a, b, x, y)= \tilde{p}(a, x, y)\tilde{p}(b| a, x, y)$ on  $A, B, X, Y$, and that $A, X, Y, \bF$ simultaneously exist, so there is a corresponding c-q state. The last constraint means that there exists a measurement on $\bF$ chosen according to $Y$ which generates $B$. Unlike the classical case, we cannot expect this constraint to be in the Markov chain form of $B-\textbf{F}Y-XA$ paralleling the classical case, since $B$ and $\bF$ do not simultaneously exist.

The following theorem is analogous to Theorem~\ref{thm:yassaee} in the presence of entanglement.

\begin{theorem} \label{thm:q-region}
$\mS_0 \subseteq \mR^q \subseteq \bigcap_{\epsilon>0} \mS_{\epsilon}.$
\end{theorem}

\begin{remark} Note that the above theorem characterizes the rate region $\mathcal R^q$ in a single-letter form. The inclusion $\mathcal S_0\subseteq \mathcal R^q$ is the achievability part and $\mathcal R^q\subseteq \bigcap_{\epsilon>0} \mathcal S_\epsilon$ is the converse. Unlike the classical case, these two bounds do not match since we do not have a dimension bound on $\bF$; if similar to the last constraint in~\eqref{eq:c-c}, we had a bound on the dimension of $\bF$, then by a simple compactness argument we could have concluded that $\mS_0=\bigcap_{\epsilon>0}\mS_\epsilon = \mR^q$.
\end{remark}

The proof of the achievability part of this theorem is based on a remote state preparation protocol with \emph{side information}.
This generalization of remote state preparation protocols could be of independent interest. Roughly speaking given the Markov chain condition $\textbf{F}-X-Y$, Alice after receiving $x$ by sending classical bits at rate $I(X;\bF|Y)$ can prepare a copy of $\bF$ at Bob's side. Then by the last constraint in \eqref{eq:f-q}, Bob by applying a measurement on $\bF$ can generate $B$. Moreover, the remote state preparation protocol is in such a way that Alice herself, has a purification of $\bF$ in hand. Then using $A-\textbf{F}X-Y$ it is shown that she can output $A$ by measuring this purification.
The converse part of the theorem follows from similar steps as the converse part of Theorem~\ref{thm:yassaee}.
For a full proof of this theorem see Appendix~\ref{AppendixProofCorrelation}.

Unfortunately Theorem~\ref{thm:q-region} does not provide an algorithm to \emph{compute} the rate region $\mathcal R^q$ since
we have no bound on the dimension of the auxiliary register $\textbf{F}$. In Appendix \ref{sec:CHSHExample} we use  an ad hoc technique based on the principle of \emph{Information Causality}~\cite{IC} to derive bounds for a particular instance of the problem, namely the CHSH-type correlations.  This  illustrate the possibility of an entirely different approach for proving computable outer bounds, when proving dimension bounds on the size of auxiliary registers is difficult. By an entirely different approach we mean a different approach than the standard converses based on identification of auxiliaries using past and future of variables (where we start from the $n$-letter form of an expansion and use, for instance the chain rule or the Csiszar sum lemma, to derive single letter forms of the expression). Although it is not clear how to extend the result to non-CHSH-type correlations, we would like to highlight the possibility of getting around dimension bounds if one is only interested in computable outer bounds.

\section{Quantum conditioning}
\label{sec:QC}\label{sec:new-tools}

In the previous section we see that rate regions involving an auxiliary quantum register naturally appear in quantum information theory especially in classical communication settings with shared entanglement. To compare such rate regions with their completely classical counterparts, we need to compare quantum auxiliary registers with classical auxiliary random variables. In particular, in the Gray-Wyner problem we see that if we could replace conditioning on a quantum register with a classical one, then entanglement does not help. So in this section we focus on understanding of quantum conditioning (entropic quantities conditioned on a quantum system). This notion appears, at least in the classical case, when we want to prove dimension bounds.

In the classical world, conditioning on a random variable has several meanings one of which is convexification. This interpretation is the crux of the Carath\'{e}odory theorem, the main classical tool for bounding cardinality of auxiliary random variables. To isolate this interpretation of conditioning we begin by some notations.
Fix finite sets $\mathcal{X}_1, \dots, \mathcal{X}_m$, and consider the map
$$p(x_1, x_2, \dots x_m)\mapsto \big(H(X_1), H(X_2), \dots, H(X_m)\big).$$
The domain of this map is the probability simplex on $\mathcal{X}_1\times\dots\times \mathcal{X}_m$. Let $\mathcal{G}$ be the graph of this map, i.e.,
\begin{align*} \mathcal{G}=\bigg\{ \big(p(x_1, \dots, x_m), &H(X_1), \dots, H(X_m)\big)  \text{for all }  p(x_1,\dots x_m)      \bigg\} . \end{align*}
Each point in $ \mathcal{G}$ corresponds to a joint distribution $p(x_1,\dots x_m)$. Then $\ConvHull(\mathcal{G})$, the convex envelope of $\mathcal{G}$ can be seen to be equal to the following:\footnote{
To compute the convex envelope of $\mathcal{G}$, one has to take a set of $k$ points and compute their weighted average (with non-negative weights adding up to one). Since each point corresponds to a joint distribution, we can think of $k$ distributions $p_i(x_1, \dots, x_m),~i=1,\dots, k$, and $k$ weights $\omega_1$, $\omega_2, \cdots, \omega_k$. Let $C$ be a random variable that takes the value $C=i$ with probability $\omega_i$ for $i=1,\dots, k$. Further define $p(x_1, \dots, x_m|C=i):=p_i(x_1, \dots, x_m)$. Then
\begin{align*}
p(x_1, \dots, x_m)&=\sum_{i}p(C=i)p(x_1, \dots, x_m|C=i)\\&=\sum_{i}\omega_ip_i(x_1, \dots, x_m),
\end{align*}
is the weighted average of $p_i(x_1, \dots, x_m)$. Further
\begin{align*}H(X_j|C)&=\sum_{i}p(C=i)H(X_j|C=i)\\&=\sum_{i}\omega_i H_{p_i(x_j)}(X_j),\end{align*}
is also a weighted average. Thus $\big(p(x_1, \dots, x_m), H(X_1\vert C), \dots, H(X_m\vert C)\big)$ is the weighted average of the set of points we started with.}
\begin{align*}
\ConvHull(\mathcal{G}) = \bigg\{ & \big(p(x_1, \dots, x_m), H(X_1\vert C), \dots, H(X_m\vert C)\big) \text{ for all } p(x_1, \dots, x_m, c)     \bigg\}.
\end{align*}
Thus conditioning on a (classical) random variable is equivalent to convexification.\footnote{Size of the alphabet of the auxiliary variable $C$ corresponds to the number of points we need to take to compute the convex hull of a set. Carath\'{e}odory's theorem is a theorem in convex geometry that bounds this number. According to this theorem, any point in the convex hull of a set in $\mathbb{R}^d$ can be expressed as the convex combination of at most $d+1$ points in the set.}

Now the question is what happens when we allow $C$ to be a quantum register. In other words what we can say about the following set
\begin{align*}
\text{QConvHull}&(\mathcal{G}):=\bigg\{  (p(x_1, \dots, x_m), H(X_1\vert \bF), \dots, H(X_m\vert \bF)) \text{ for all } p(x_1, \dots, x_m) \text{ and } \rho^{\bF}_{x_1\dots x_m}     \bigg\}.
\end{align*}
Observe that $\text{QConvHull}(\mathcal{G})$ is convex and contains $\ConvHull(\mathcal{G})$. The question is whether this containment is strict or equality holds. One difficulty of understanding $\text{QConvHull}(\mathcal{G})$ is that unlike the classical case, no bound on the dimension of $\bF$ is known. This means that we do not even know how to compute $\text{QConvHull}(\mathcal{G})$. \\

\noindent\textbf{Example.} Let $m=3$ and $X_3=(X_1, X_2)$. Then for any $p(x_1,x_2)$ the coordinates of the triple $\big(H(X_1|\bF),H(X_2|\bF),H(X_3|\bF)\big)$ satisfy $H(X_3|\bF)\leq H(X_1|\bF)+H(X_2|\bF)$. Suppose we are interested in the set of triples where equality $H(X_3|\bF)=H(X_1|\bF)+H(X_2|\bF)$ holds. In this case $I(X_1;X_2|\bF)=0$. Now using the structure of states that satisfy strong subadditivity
of quantum entropy with equality~\cite{Haydenetal04} we conclude that
there exists a classical random variable $C$ such that $I(X_1;X_2|C)=0$, $H(X_1|C)=H(X_1|\bF)$ and $H(X_2|C)=H(X_2|\bF)$. This means that the quantum and classical regions are the same under the constraint that the third coordinate of the triple is equal to the sum of the first two coordinates.\\

The first main result of this section is that quantum conditioning is \emph{strictly richer} than classical conditioning. Here we introduce new tools that could be useful in bounding the dimension of quantum registers as well.

\begin{theorem}\label{thm:QCMain1}

\begin{enumerate}
                 \item[\rm{(a)}] The following three statements are \emph{equivalent}:
                 \begin{enumerate}
                                        \item[\rm{1.}] $\text{\emph{QConvHull}}(\mathcal{G})=\text{\emph{ConvHull}}(\mathcal{G})$ for any finite sets $\mathcal{X}_1, \dots, \mathcal{X}_m$.

                                        \item[\rm{2.}] For a classical-quantum channel $X\rightarrow \bF$ determined by a collection of density marices $\rho^{\bF}_{x}$ for $x\in \mathcal{X}$, consider the function $p(x)\mapsto I(X;\bF)$ for distributions $p(x)$ on $\mathcal{X}$. Then for every $\epsilon>0$ there exists a classical channel $X\rightarrow C$ determined by $q(c|x)$ such that $\big|I(X;\bF)-I(X;C)\big|\leq \epsilon$ for all $p(x)$.

                                         \item[\rm{3.}] For any arbitrary $q(x,y,z)$ consider the optimization problem
\[\sup_{\textbf{F}-X-YZ}I(\textbf{F};Y)-I(\textbf{F};Z),
\]
over all quantum registers $\textbf{F}$ satisfying $\bF-X-YZ$. Then the supremum is a maximum and is attained at a classical $\bF$.
                                                                          \end{enumerate}
                 \item[\rm{(b)}] There is a counterexample for part {\rm{(1)}} implying that all of the above three statements are false.
               \end{enumerate}
\end{theorem}

 Part (a2) of the theorem introduces the problem of uniformly approximating the mutual information curve (or surface) $p(x)\mapsto I(X;\bF)$ with classical ones. The mutual information $I(X;\bF)$ is concave in $p(x)$. Therefore the problem is that of approximating a concave curve (or surface) with another one. Statement of this theorem says that this is not possible for some $\{\rho^{\bF}_{x}\}$.

Consider the optimization problem introduced in part (a3) of the theorem. Because of the Markov chain condition, the auxiliary system $\bF$ is determined by the collection of states $\{ \sigma_x:\; x\in \mathcal{X}\}$. Note that the supremum is not known to be computable because no bound on the dimension of $\bF$ is known. The classical form of the expression is $\sup_{p(c|x)}I(C;Y)-I(C;Z)$, where we are taking the supremum over all classical channels $p(c|x)$. In the classical case we know that the supremum is indeed a maximum, and further the cardinality of $\mathcal{C}$ can be bounded from above by $|\mathcal{X}|$ using the Convex Cover Method (see \cite[Appendix C]{AbbasYoungHan}) where $g_j$, $j-1,2,\cdots, |\mathcal{X}|$ in the statement of the lemma are defined as follows:
$$g_j(\pi)=\left\{
  \begin{array}{ll}
    \pi(j), & j=1,2,\cdots, |\mathcal{X}|-1; \\
    H(Y)-H(Z), & j=|\mathcal{X}|.
  \end{array}
\right.
$$
Observe that it is because of the Markov chain $C-X-YZ$ that the cardinality bound would not depend on $|\mathcal{Y}|$ and $|\mathcal{Z}|$. The above theorem shows that there exists a distribution $q(x,y, z)$ such that
\begin{align}\label{eq:st2}
\sup_{\textbf{F}-X-YZ}&I(\textbf{F};Y)-I(\textbf{F};Z) > \max_{p(c|x)}I(C;Y)-I(C;Z).
\end{align}

In the second main result in this section, we consider the above statements in the special case of $|\mathcal{X}|=\dim\bF=2$.
\begin{theorem}\label{thm:QCMain2} The following two statement hold:

\begin{enumerate}
\item Let $|\mathcal{X}|=\dim \bF=2$ and consider a channel $X\rightarrow \bF$ determined by  $\rho^{\bF}_{0}, \rho^{\bF}_{1}$. Then one can find a classical channel $q(c|x)$ such that $I(X;\bF)=I(X;C)$ for all $p(x)$.
                  \item For an arbitrary $q(x,y,z)$ where $X$ is binary, consider the supremum
\begin{align*}
\sup_{\textbf{F}-X-YZ}I(\textbf{F};Y)-I(\textbf{F};Z),
\end{align*}
over all quantum registers $\textbf{F}$ \emph{of dimension two}. Then the supremum is a maximum and one can always find a maximizer $\textbf{F}$ that is classical.
\end{enumerate}
\end{theorem}

We leave a prove of Theorem~\ref{thm:QCMain2} for Appendix~\ref{app:carve-d2}. Theorem~\ref{thm:QCMain1} is proved in the following two subsections.

\subsection{Proof of part (a) of Theorem~\ref{thm:QCMain1}}

We show that (a1) implies (a3), (a2) implies (a1), and (a3) implies (a2). The fact that (a1) implies (a3) is immediate noting that
\begin{align*}
I(\textbf{F};Y)-I(\textbf{F};Z)&=H(Y)-H(Z) -H(Y|\textbf{F})+H(Z|\textbf{F}),
\end{align*}
can be expressed in terms of conditional entropies given a quantum register. The Markov chain constraint $\textbf{F}-X-YZ$ can also be written as $H(YZ|X)=H(YZ|X,\textbf{F})$ or alternatively as $H(YZ|X)=H(YZX|\textbf{F})-H(X|\textbf{F})$ in terms of conditional entropies given the quantum register.

To show that (a2) implies (a1), take some arbitrary finite sets $\mathcal{X}_1, \dots, \mathcal{X}_m$, distribution $q(x_1,\cdots,x_m)$ and states $\rho^{\bF}_{x_1,\cdots,x_m}$. Let $X=(X_1,\cdots, X_m)$. Then by (a2) for any $\epsilon>0$, one can find a classical channel $q(c|x)$ such that $\big|I(X;\bF)-I(X;C)\big|\leq \epsilon$, \emph{for all} $p(x)$. We show that this implies that $\big|I(X_i;\bF)-I(X_i;C)\big|\leq \epsilon$ for $1\leq i\leq m$. Observe that
\begin{align*}
I(X_i;\bF)&=I(X_1X_2\cdots X_m;\bF)-I(X_1X_2\cdots X_m;\bF|X_i)\\&=I(X;\bF)-I(X;\bF|X_i)\\&=I(X;\bF)-\sum_{x_i}p(x_i)I(X;\bF|X_i=x_i).
\end{align*}
But $I(X;\bF|X_i=x_i)$ is nothing but the mutual information between $X$ and $\bF$ at the conditional distribution $p_{X|X_i}(x|x_i)$. Thus we have $\big|I(X;\bF)-I(X;C)\big|\leq \epsilon$, and $\big|I(X;\bF|X_i=x_i)-I(X;C|X_i=x_i)\big|\leq \epsilon$ for any $x_i$. This implies
$$\big|I(X_i;\bF)-I(X_i;C)\big|\leq \epsilon+\sum_{x_i}p(x_i)\epsilon=2\epsilon,$$
and
$$\big|H(X_i|\bF)-H(X_i|C)\big|\leq 2\epsilon.$$
Therefore we have approximated $H(X_i|\bF)$ with a classical $H(X_i|C)$ for all $i$ within $2\epsilon$. This implies that $$\text{QConvHull}(\mathcal{G})\subset\overline{\ConvHull(\mathcal{G})},$$
where $\overline{\ConvHull(\mathcal{G})}$ is the closure of $\ConvHull(\mathcal{G})$. But $\ConvHull(\mathcal{G})$ is a closed set since in the classical case we can bound the cardinality of the auxiliary $C$ with $|\mathcal{X}_1|\times|\mathcal{X}_2|\times|\mathcal{X}_3|\cdots |\mathcal{X}_m|+m$ using Carath\'{e}odory theorem.

Showing that (a3) implies (a2) is challenging and contains the main ideas of this section. Assume that for all distributions $q(x,y, z)$ we have
\begin{align}\label{eq:st2}
\sup_{\textbf{F}-X-YZ}&I(\textbf{F};Y)-I(\textbf{F};Z)=  \max_{p(c|x)}I(C;Y)-I(C;Z).
\end{align}
Note that the right hand side is computable since we can impose the restriction $|\mathcal{C}|\leq |\mathcal{X}|$. Let $\mathcal{P}$ denote the class of all classical channels $p(c|x)$ with $|\mathcal{C}|\leq |\mathcal{X}|$.\\

\noindent
\textbf{Rexpressing in terms of a max-min equation:}
Fix a distribution $q(x)$ on $\mathcal{X}$ and an arbitrary c-q channel $X\rightarrow \textbf{F}$ with states $\sigma_x,~x\in \mathcal{X}$. Without loss of generality we assume $q(x)>0$ for all $x\in \mathcal{X}$.
By our assumption~\eqref{eq:st2} we have
\begin{align*}
I(\textbf{F};Y)&-I(\textbf{F};Z)\leq \max_{p(c|x)\in \mathcal{P}}I(C;Y)-I(C;Z),~~~\forall q(y,z|x).
\end{align*}
Equivalently,
\begin{align*}
\max_{q(y,z|x)}\bigg[&I(\textbf{F};Y)-I(\textbf{F};Z)-\max_{p(c|x)\in \mathcal{P}}I(C;Y)-I(C;Z)\bigg]\leq 0.
\end{align*}
Take an arbitrary $\epsilon>0$. Lemma \ref{lemmaCard1} of Appendix~\ref{app:details} shows that one can find a finite set of channels $p(c_j|x), (j=1,2,\cdots,M_{\epsilon})$ to uniformly approximate the continuous function $I(C;Y)-I(C;Z)$ over the compact set $\mathcal{P}$ within $\epsilon$ (an $\epsilon$-net); i.e., for every $p(c|x)\in \mathcal{P}$ there exists $j$ such that
\begin{align*}
\bigg|\big(&I(C;Y)-I(C;Z)\big)-\big(I(C_j;Y)-I(C_j;Z)\big)\bigg|\leq \epsilon,~~\forall q(y,z|x).
\end{align*}
Thus
\begin{align*}\max_{q(y,z|x)}\bigg[I(&\textbf{F};Y)-I(\textbf{F};Z)-\max_{1\leq j\leq M_\epsilon}I(C_j;Y)-I(C_j;Z)\bigg]\leq \epsilon.\end{align*}
This means that
\begin{align*}&\max_{q(y,z|x)}\bigg[I(\textbf{F};Y)-I(\textbf{F};Z)-\max_{\lambda_j\geq 0: \, \sum_{j}\lambda_j=1}\sum_{j}\lambda_j\big(I(C_j;Y)-I(C_j;Z)\big)\bigg]\leq \epsilon.
\end{align*}
Alternatively,
\begin{align}\label{eq:maxmineq}
\max_{q(y,z|x)}\,\, \min_{\lambda_j\geq 0:\, \sum_{j}\lambda_j=1}\bigg[I(\textbf{F};Y)-I(\textbf{F};Z) -\sum_{j}\lambda_j\big(I(C_j;Y)-I(C_j;Z)\big)\bigg]\leq \epsilon.
\end{align}\\

\noindent
\textbf{Exchanging the order of min and max:}
The main step in the proof is to use Lemma \ref{lemmaCard2} of Appendix~\ref{app:details} to exchange the order of maximum and minimum to get
\begin{align}\label{eq:minmaxeq}
\min_{\lambda_j\geq 0:\sum_{j}\lambda_j=1}\,\, \max_{q(y,z|x)}\bigg[I(\textbf{F};Y)-I(\textbf{F};Z)-\sum_{j}\lambda_j\big(I(C_j;Y)-I(C_j;Z)\big)\bigg]\leq \epsilon.
\end{align}
Thus there exists a choice of $\lambda_j$, \emph{not depending on $q(y,z|x)$}, such that
\begin{align*}
&\max_{q(y,z|x)}\bigg[I(\textbf{F};Y)-I(\textbf{F};Z) -\sum_{j}\lambda_j\big(I(C_j;Y)-I(C_j;Z)\big)\bigg]\leq \epsilon,
\end{align*}
which is equivalent to
\begin{align*} 
\max_{q(y,z|x)}&\left[H(\textbf{F}|Z)-H(\textbf{F}|Y) -\sum_{j}\lambda_j\big(H(C_j|Z)-H(C_j|Y)\big)\right]\leq \epsilon.
\end{align*}

\noindent
\textbf{Implication of the max-min exchange:}
Assuming $q(y, z|x) = q(y|x)q(z|x)$ we obtain
\begin{align*}
\max_{q(z|x)}\left[H(\textbf{F}|Z)-\sum_{j}\lambda_j H(C_j|Z)\right] +\max_{q(y|x)}\left[-H(\textbf{F}|Y)+\sum_{j}\lambda_j H(C_j|Y)\right]\leq \epsilon.
\end{align*}
We can express the two maximums in terms of the same channel as follows
\begin{align}\label{eq:min-max-ineq}
\max_{q(z|x)} \left[H(\textbf{F}|Z)-\sum_{j}\lambda_j H(C_j|Z)\right]\leq \epsilon+\min_{q(z|x)}\left[H(\textbf{F}|Z)-\sum_{j}\lambda_j H(C_j|Z)\right].
\end{align}

Let us define
\begin{align*}W(p(x))&=H(\textbf{F})-\sum_{j}\lambda_j H(C_j)\\&
=H\left(\sum_x p(x)\sigma_x\right)-\sum_{j}\lambda_j H\left(\sum_x p(x)p(c_j|x)\right).
\end{align*}
Then the left hand side of \eqref{eq:min-max-ineq}
is the upper concave envelope\footnote{The upper concave envelope of a function $f(t)$ is the smallest \emph{concave} function $g(t)$ such that $f(t)\leq g(t)$ for all $t$. Lower convex envelope is defined similarly.} of the graph of $W(p(x))$ whereas the right hand side
is the lower convex envelope of $W(p(x))$. We know that the difference between the two is at most $\epsilon$. Were these two are exactly equal, the function $W(p(x))$ must have been linear in $p(x)$ for \emph{all} $p(x)$ (and not just the $q(x)$ we started with). Therefore the function $W(p(x))$ is almost linear.

The function
\begin{align*} V(p(x))&=I(\textbf{F};X)-\sum_{j}\lambda_jI(C_j;X)\\&=I(\textbf{F};X)-I(D;X)\end{align*}
is equal to $W(p(x))$ plus a linear term in $p(x)$, where here $D$ is defined as $D=(U,C_U)$ where $U$ is a random variable, independent of $X$ taking value $j$ with probability $\lambda_j$. As a result, the upper concave envelope and lower convex envelope of $V(p(x))$ at $q(x)$ are also $\epsilon$-close to each other.

The function $V(p(x))$ is zero when $p(x)$ assigns probability one to a single symbol (i.e. on the vertices of the probability simplex). Thus its lower convex envelope is less than or equal to the zero function, whereas its upper concave envelope is greater than or equal to zero. Since the gap between the two is at most $\epsilon$ at the given $q(x)$ and $q(x)>0$ for all $x$,  $|V(p(x))|$ should be close to zero for every $p(x)$. Thus $$\big|I(\textbf{F};X)-I(D;X)\big|\leq O(\epsilon), \qquad \forall p(x).$$ This completes the proof.

\subsection{Proof of part (b) of Theorem~\ref{thm:QCMain1}}
\begin{figure*}
\begin{align*}
\begin{matrix}
           &\vline & 1& 2 & 3 & 4\\
             \hline
\alpha & \vline &(1, 0,0,0) & (0, 1,0,0) & (0, 0,1,0) & (0, 0,0,1)\\
\beta & \vline &(0, 1,1,0) & (1, 0,0,-1) & (1, 0,0,1) & (0, 1,-1,0)\\
\gamma & \vline &(1, 1,1,1) & (1, -1,1,-1) & (1, -1,-1,1) & (1, 1,-1,-1)\\
\delta & \vline &(1, -1,0,0) & (1, 1,0,0) & (0, 0,1,1) & (0, 0,1,-1)\\
\epsilon & \vline &(-1, 1,1,1) & (1, 1,1,-1) & (1, -1,1,1) & (1, 1,-1,1)\\
\zeta & \vline &(1, 0,1,0) & (0, 1,0,1) & (1, 0,-1,0) & (0, 1,0,-1)\\
\end{matrix}
\end{align*}
\caption{Definition of vectors needed for the proof of part (b) of Theorem~\ref{thm:QCMain1}.}
\label{fig:nadd}
\end{figure*}
The counterexample is inspired by the examples of (classical) channels for which the one-shot entanglement-assisted zero-error capacity is greater than the zero-error capacity~\cite{Zero1, Zero2}. Here we explain the details based on the Kochen-Specker type channel of \cite{Zero1}. For more information on Kochen-Specker sets and their importance see~\cite[Chapter 7]{Peres}. See also~\cite{MSS} for a generalization of Kochen-Specker sets. 

Let $\mathcal{M}=\{\alpha, \beta, \gamma, \delta, \epsilon, \zeta\}$ and $\mathcal{X}=\{\theta_i:\,\, \theta\in \mathcal{M}, 1\leq i\leq 4\}$. Moreover, let $\mathcal{Y} = \{ S_1, S_2,\dots, S_{18}  \}$ where $S_i$'s are certain four-elements subsets of $\mathcal{X}$:
\begin{align*}
S_1&=\{\alpha_1, \alpha_4, \beta_1, \beta_4\}, \quad S_2=\{\gamma_1, \gamma_4, \delta_1, \delta_4\}, \quad S_3=\{\epsilon_1, \epsilon_4, \zeta_1, \zeta_4  \},\\
 S_4&=\{\alpha_2, \alpha_3, \beta_2, \beta_3  \}, \quad S_5=\{ \gamma_2, \gamma_3, \delta_2, \delta_3  \}, \quad
S_6 =\{\epsilon_2, \epsilon_3, \zeta_2, \zeta_3   \},\\
S_7&=\{\alpha_1, \alpha_3, \zeta_2, \zeta_4  \}, \quad S_8=\{ \beta_2, \beta_4, \gamma_1, \gamma_3  \}, \quad
S_9 =\{\delta_2, \delta_4, \epsilon_1, \epsilon_3  \},\\
 S_{10}&=\{\alpha_2, \alpha_4, \zeta_1, \zeta_3  \}, \quad S_{11}=\{ \beta_1, \beta_3, \gamma_2, \gamma_4  \}, \quad
S_{12} =\{\delta_1, \delta_3, \epsilon_2, \epsilon_4  \},\\
S_{13}&=\{\alpha_1, \alpha_2, \delta_3, \delta_4  \}, \quad S_{14}=\{ \beta_1, \beta_2, \epsilon_3, \epsilon_4  \}, \quad
S_{15} =\{\gamma_1, \gamma_2, \zeta_3, \zeta_4  \},\\
S_{16}&=\{\alpha_3, \alpha_4, \delta_1, \delta_2  \}, \quad S_{17}=\{ \beta_3, \beta_4, \epsilon_1, \epsilon_2  \}, \quad
S_{18} =\{\gamma_3, \gamma_4, \zeta_1, \zeta_2  \}.
\end{align*}
Finally let $\bF$ be a $4$-level quantum system and for $\theta_i\in \mathcal{X}$ define $\rho_{\theta_i}^{\bF} = |\psi_{\theta_i}\rangle\langle \psi_{\theta_i}|$ where $|\psi_{\theta_i}\rangle$'s are proportional to the vectors given in Fig. \ref{fig:nadd}

Note that $\vert \psi_{\theta_i}\rangle$ and $\vert \psi_{\theta'_j}\rangle$ are orthogonal if and only if $\theta=\theta'$ or there exists $k$, $1\leq k\leq 18$, such that $\theta_i, \theta'_j\in S_k$. In fact each of the $18$ subsets
$\{  \vert \psi_{\theta_i}\rangle : \,\,\theta_i\in S_k \}$ for all $k$, as well as the $6$ subsets $\{\vert \psi_{\theta_i}\rangle: \,\, i=1,\dots, 4  \}$ for all $\theta\in \mathcal{M}$, consist an orthonormal basis for the Hilbert space of $\bF$. To represent these orthogonality relations form a graph on the vertex set $\mathcal{X}$ and connect two vertices $\theta_i$ and $\theta'_{j}$ if $\langle \psi_{\theta_i}\vert \psi_{\theta'_j}\rangle =0$. This orthogonality graph contains $18+6$ cliques corresponding to the above orthonormal bases, and the edge set of the graph is the union of these cliques. The independence number of this graph, namely the maximum number of vertices no two of which are connected, is $5$.

Now consider the following distribution on $XMY$. Let $p(\theta_i)=\frac{1}{24}$ be the uniform distribution on $X$. The distribution on $M$ is $p(\theta'| X=\theta_i) =1$ iff $\theta'=\theta$. To define the distribution on $Y$ note that for each $\theta_i\in \mathcal{X}$, there are exactly three indexes $k$ such that $\theta_i\in S_k$. Let $p(S_k| X=\theta_i)=1/3$ iff $\theta_i\in S_k$. Observe that the one-shot zero-error capacity of the channel $X\rightarrow Y$ (determined by $p(S_k| X=\theta_i)$) is $\log 5$ because the independence number of the orthogonality graph is $5$ (see \cite{Zero1} for more details).
We finally define the state of $\bF$ to be $\rho_{\theta_i}$ when $X=\theta_i$.

Now it is easy to verify that
\[ H(X\vert \bF)= \log 6, \quad H(M\vert \bF) = H(M) = \log 6, \quad H(Y\vert \bF) = H(Y) = \log 18.  \]
These equations are all based on the fact that the average of states $\rho_{\theta_i}$ when $\theta_i$ ranges over a clique of the orthogonality graph, is equal to the maximally mixed state.
So by the above notation
$$(p(x, \theta, y), \log 6, \log 6, \log 18) \in \text{QConvHull}(\mathcal{G}),$$
where here $n=3$ and $X_1=X$, $X_2=M$ and $X_3=Y$. To proof $\text{QConvHull}(\mathcal{G}) \neq \ConvHull(\mathcal{G})$ we show that this point does not belong to $\ConvHull(\mathcal{G})$. Suppose there exists a classical random variable $C$ such that
$$ H(X\vert C)= \log 6, \quad H(M\vert C) = H(M) = \log 6,\quad H(Y\vert C) = H(Y) = \log 18. $$
The above three equations imply that (the proof comes later)
\[   I(C; M) =0, \quad  H(X\vert CM)=0 ,\quad MC- X- Y, \quad H(M\vert CY)=0.\]
Pick a $c$ such that $p(C=c)\neq 0$ and consider the distribution $p(x, \theta , y| C=c)$. By the first equation $p(\theta|C=c) = p(\theta) = 1/6$, and by the second equation $X$ is deterministically computed from $M$ (and $C=c$). Using the structure of the distribution $p(x, \theta)$ we find that for every $\theta\in \{\alpha, \beta , \gamma, \delta, \epsilon, \zeta  \}$ there exists $i$, $1\leq i\leq 4$,  such that $p(x, \theta | C=c) = 1$ iff $x=\theta_i$. We denote the set of these six $\theta_i$ by $T$. Thus $|T|= 6$ and for every $\theta\in \mathcal{M}$ there exists $i$ such that $\theta_i\in T$. $MC-X-Y$ implies that $p(y| C=c, M=\theta, X=\theta_i) = p(y| X=\theta_i).$ So $y$ is uniformly distributed among the three subsets $S_k$ that contain $\theta_i$. Finally the last equation says that $y$ (and $C=c$) uniquely determines $\theta$. This means that, there is no $S_k$ that contains more that two elements of  $T$. As a result, $T$ is an independent set of the orthogonality graph of size $6$. This is a contradiction since the independence number of this graph is $5$.

A more intuitive argument is based on a zero-error communication protocol over the channel $X\rightarrow Y$ using $C$ as \emph{shared randomness}. We take $M$ as the message to be transmitted from the sender to the receiver. Note that by the first equation $M$ is independent of $C$ (the shared randomness), so this analogy makes sense.
$H(X\vert CM)=0$ implies that $X$ is a function of $CM$. So the sender computes $X$ from $M$ and $C$, and sends it over the channel. By $MC-X-Y$ given the input, the output of the channels is independent of $M$ and $C$. Finally the last equation means that the receiver can decode $M$ from $Y$ and $C$, and this can be done with \emph{no error}. As a result the one-shot zero-error capacity of $X\rightarrow Y$ is at least $H(M) = \log 6$ which is a contradiction.

 We finish this section by proving that the three equations
$$ H(X\vert C)= \log 6, \quad H(M\vert C) = H(M) = \log 6,\quad H(Y\vert C) = H(Y) = \log 18, $$
imply
\[    I(C; M) =0, \quad  H(X\vert CM)=0, \quad MC- X- Y, \quad H(M\vert CY)=0.\]
The first equation directly follows from $H(M\vert C) = H(M) = \log 6$. To show the third and last equations we write
\begin{align*}  H(X\vert CY)&=H(XY\vert C)-H(Y\vert C)\\&= H(X\vert C)+H(Y\vert CX)-H(Y)\\&= H(X\vert C)+H(Y\vert MCX)-H(Y)\\&\leq H(X\vert C)+H(Y\vert X)-H(Y)\\&=
\log 6+\log 3-\log 18=0,
\end{align*}
where in the third line we have used the fact that $M$ is uniquely determined in terms of $X$. The above equation implies that $H(M\vert CY)=0$. Furthermore, the inequality should hold with equality. This gives us the constraint $MC- X- Y$.
To show the second identity we write
\begin{align*}H(X\vert CM)&=H(XM\vert C)-H(M\vert C)
\\&=H(X\vert C)-H(M\vert C)
\\&=\log 6-\log 6=0,\end{align*}
where we have again used the fact that $M$ is a function of $X$.


\section{Dimension bounds and the mutual information curve}\label{sec:DB}
Let us consider the optimization problem
$$\sup_{\textbf{F}-X-YZ}I(\textbf{F};Y)-I(\textbf{F};Z),$$
and assume that we can restrict to quantum registers with dimension bound $d^*$ that universally works for all $p(x,y,z)$. That is, assume that there is a bound $d^*$ that depends only on $|\mathcal{X}|,|\mathcal{Y}|,|\mathcal{Y}|$ such that in the above optimization we may restrict the dimension of $\bF$ to be $d^*$.
This is the problem considered in Section \ref{sec:introtool}.

Fix a distribution $q(x)$ on $\mathcal{X}$ and an arbitrary c-q channel $X\rightarrow \textbf{F}$ which maps $x\in \mathcal X$ to $\sigma_x$. Without loss of generality we assume $q(x)>0$ for all $x\in \mathcal{X}$.
By our assumption we have
\begin{align*}&I(\textbf{F};Y)-I(\textbf{F};Z)\leq \max_{\textbf{E}: \dim(\textbf{E})\leq d^*}I(\textbf{E};Y)-I(\textbf{E};Z),~~~~\forall q(y,z|x),\end{align*}
where the maximum is taken over all $\textbf{E}$ with the given dimension bound that satisfy $\textbf{E}-X-YZ$ (indeed over all $\rho^{\bE}_{x}, x\in \mathcal{X}$ where $\rho^{\bE}_{x}$ is of dimension $d^*$ for all $x$). Equivalently,
\begin{align*}
&\max_{q(y,z|x)}\bigg[I(\textbf{F};Y)-I(\textbf{F};Z) -\max_{\textbf{E}: \dim(\textbf{E})\leq d^*}I(\textbf{E};Y)-I(\textbf{E};Z)\bigg]\leq 0.
\end{align*}
Take an arbitrary $\epsilon>0$. Using a similar argument as in Lemma \ref{lemmaCard1} of Appendix~\ref{app:details}, we can discretize the set of all $\textbf{E}$ with the given dimension bound that satisfy $\textbf{E}-X-YZ$, i.e. to find a finite set of $\mathbf{E}_j$ with $\rho^{\bE_j}_{xj},~x\in \mathcal{X}, (j=1,2,\cdots,M_{\epsilon})$ to uniformly approximate the continuous function $I(\textbf{E};Y)-I(\textbf{E};Z)$ over the compact set of all such auxiliary $\textbf{E}$ with the given dimension bound within $\epsilon$ (an $\epsilon$-net).
Thus
\begin{align*}
&\max_{q(y,z|x)}\bigg[I(\textbf{F};Y)-I(\textbf{F};Z)-\max_{1\leq j\leq M_\epsilon}I(\mathbf{E}_j;Y)-I(\mathbf{E}_j;Z)\bigg]\leq \epsilon.
\end{align*}
This means that
\begin{align*}
&\max_{q(y,z|x)}\bigg[I(\textbf{F};Y)-I(\textbf{F};Z)-\max_{\lambda_j\geq 0: \, \sum_{j}\lambda_j=1}\sum_{j}\lambda_j\big(I(\mathbf{E}_j;Y)-I(\mathbf{E}_j;Z)\big)\bigg]\leq \epsilon.
\end{align*}
Alternatively,
\begin{align}\label{eq:maxmineq}
&\max_{q(y,z|x)}\,\, \min_{\lambda_j\geq 0:\, \sum_{j}\lambda_j=1}\bigg[I(\textbf{F};Y)-I(\textbf{F};Z)-\sum_{j}\lambda_j\big(I(\mathbf{E}_j;Y)-I(\mathbf{E}_j;Z)\big)\bigg]\leq \epsilon.
\end{align}
Following very similar arguments as before and exchanging the order of max and min, we conclude that
$$\big|I(\textbf{F};X)-\sum_{j}\lambda_jI(\textbf{E}_j;X)\big|\leq O(\epsilon), \qquad \forall p(x),$$
for a choice of $\lambda_j$'s independent of $q(y, z|x)$.
This equation suggests that to study dimension bounds, it would help to understand the behavior of the concave function $p(x)\mapsto I(\textbf{F};X)$ in terms of the dimension of $\bF$. Indeed the question of finding dimension bounds is reduced to the question of whether or not these functions become more and more complicated (in structure) as we increase the dimension, so that they cannot be written in terms of those functions with smaller dimensions.


\section{Conclusion}
In this paper we presented a tool for proving dimension bounds on auxiliary quantum systems. Motivated by the use of Carath\'{e}odory theorem in classical information theory, we formalized the problem of quantum conditioning. We showed that quantum conditioning coincides with (classical) conditioning in a very special case but in general goes beyond it. We also studied the role of entanglement in \emph{classical} communication problems from an \emph{information theoretic} point of view. We observed that unlike the problem of point-to-point channel capacity, entanglement does help in the problem of simulation of correlations by studying CHSH-type correlations. This may not seem surprising given that shared \emph{randomness} also helps  in the problem of simulation of correlations. However, the situation is different in the Gray-Wyner problem. Given that shared randomness does not increase the capacity in the classical Gray-Wyner problem (which is a communication problem), should entanglement turns out to be helpful in this problem, it would serve as an interesting example. To pursue this further, one needs to develop new tools for bounding the dimension of auxiliary quantum systems.

In classical information theory, all of the known tools for bounding the cardinality are affirmative, meaning that they either can be used to prove a cardinality bound, or do not yield any result at all. We know of no expression where a cardinality bound is rigorously proven not to exists.
Thus an interesting problem would be to study the plausibility of an explicit example where a single auxiliary quantum register shows up, and there is no dimension bound. A good candidate could be studying the dimension of $\bF$ in equation~\eqref{eq:f-q} since here a dimension bound may not exist in general. Based on evidences from numerical simulations in~\cite{PalVertesi} it is \emph{conjectured} that the maximum violation of a particular Bell inequality (called I3322 inequality) with shared entanglement does not happen in finite dimensions. Although this conjecture is about simulating bipartite correlations in a single shot and does not directly apply into our framework in Theorem~\ref{thm:q-region}, it suggests that we may not have a dimension bound on $\bF$ in equation~\eqref{eq:f-q}.

\vspace{.2in}

\section*{Acknowledgments}
We are thankful to Andreas Winter for letting us know about his result on entanglement-assisted Gray-Wyner problem and fruitful discussions. We are also grateful to unknown referees whose comments significantly improved the presentation of the paper. SB was in part supported by National Elites Foundation and by a grant from IPM (No.\ 91810409).
AG was in part supported by a grant from IPM (No.\ CS1390-3-01).


\appendix

\vspace{.4in}
\begin{center}
\textbf{\Large Appendix}
\end{center}

\section{Some notations and useful lemmas}\label{app:useful-lem}
We frequently use the gentle measurement lemma in this paper:
\begin{lemma}\label{lem:GML} \emph{(Gentle measurement lemma \cite{Gentle})}
Let $\rho$ be a quantum state and $\{M_0, M_1\}$ be a binary measurement such that $\tr(M_0^{\dagger}M_0\rho)\geq 1-\epsilon$. If we measure $\rho$ with $\{M_0, M_1\}$ and obtain $0$ as the outcome, then the post measurement state $\rho'$ proportional to $M_0\rho M_0^{\dagger}$ satisfies
$$\|\rho - \rho'\|_{1} \leq 2\sqrt{\epsilon}.$$
\end{lemma}

We fix an orthonormal basis $\{\vert v_1\rangle, \dots, \vert v_{\dim \bF} \rangle\}$ for the Hilbert space of $\bF$ and write all the transposes ($T$) with respect to this basis. Moreover, we set
$$\vert \Phi\rangle_{\bE\bF} = \frac{1}{\sqrt{\dim \bF}}  \sum_{i=1}^{\dim \bF}  \vert v_i\rangle_{\bE}\vert v_i\rangle_{\bF},$$
to be the maximally entangled state over $\bE\bF$ where $\bE$ is a copy of $\bF$.

We assume the reader is familiar with the notions of typicality and conditional typicality (see for example~\cite[Chapter 12]{NielsenChuang} and~\cite[Chapter 14]{Wilde}). Here we only fix some notations. Throughout this paper by typicality we mean strong typicality.
The state of a c-q system $X\bF$ has the form
$$\rho_{X\bF} = \sum_x p(x)|x\rangle \langle x| \otimes \rho_x$$
and subsystem $\bF$ has the average state $\rho=\sum_x p(x)\rho_x$. The $\delta$-typical subspace of $\rho$ is determined by a projection $\Pi_{\rho, \delta}^n$ acting on the Hilbert space of $\bF^n$. We may drop the index $\rho$ in $\Pi_{\rho, \delta}^n$ when there is no confusion.

For every $\delta, \epsilon>0$ and sufficiently large $n$ we have
$$\tr \left(   \Pi_{\delta}^n \rho^{\otimes n}  \right) \geq 1- \epsilon,$$
$$ (1-\epsilon)2^{n(H(\bF)-c\delta)}\leq \tr~\Pi_{\delta}^n \leq 2^{n(H(\bF)+c\delta)},$$
and
$$2^{-n(H(\bF)+c\delta)} \Pi_{\delta}^n \leq \Pi_{\delta}^n\rho^{\otimes n}\Pi_{\delta}^n \leq 2^{-n(H(\bF)-c\delta)} \Pi_{\delta}^n,$$
where $c$ is some constant.

For a given $x^n$ we define $\rho_{x^n} = \rho_{x_1}\otimes \cdots \otimes \rho_{x_n}$. Moreover, by $\Pi_{\rho, \delta}^{x^n}$ we mean the conditional $\delta$-typical projection. Again we may denote $\Pi_{\rho, \delta}^{x^n}$ by $\Pi_{\delta}^{x^n}$ when there is no confusion. For every $\delta$-typical $x^n$ we have
$$\tr\left(\Pi_{\delta}^{x^n}\rho_{x^n}\right)\geq 1-\epsilon,$$
$$(1-\epsilon)2^{n(H(\bF|X)-\delta'')}\leq \tr~\Pi_{\delta}^{x^n} \leq 2^{n(H(\bF|X)+\delta'')},$$
and
\begin{align*}2^{-n(H(\bF|X)+\delta'')} \Pi_{\delta}^{x^n}& \leq \Pi_{\delta}^{x^n}\rho_{x^n}\Pi_{\delta}^{x^n} \leq 2^{-n(H(\bF|X)-\delta'')} \Pi_{\delta}^{x^n},\end{align*}
where $\delta''=\delta'|\mathcal{X}|\log (\dim \bF)+c\delta+|\mathcal{X}|c\delta\delta'$.  We further have
$$\tr(\Pi^n_{\delta}\rho_{x^n})\geq 1-\epsilon.$$

The following important lemmas, which might be of independent interest, will be used in the proof of Theorem~\ref{thm:q-region}.

\begin{lemma}\label{lem:purification}
Suppose $X-\bF - Y$ where $X, Y$ are classical random variables and $\bF$ is a quantum register which for every $y\in \mathcal Y$ is purified by $\bE$. Then having access to $\bE$ one can generate $X$ independent of $Y$.
\end{lemma}

\begin{proof}
Let $\tau_{X\bF Y}$ be the joint state of $X\bF Y$.
Due to the structure of tripartite states with $X-\bF-Y$, i.e., satisfy strong subadditivity with equality~\cite{Haydenetal04},
there exists a decomposition of the Hilbert space of $\bF$ of the form $\bigoplus_j \bF^j_{L}\otimes \bF^j_R$ such that
\begin{align*}\tau_{X\bF Y} = \bigoplus_j  q(j) &\left(  \sum_x  p_j(x) |x\rangle\langle x| \otimes \rho^j_x    \right)\otimes \left(   \sum_y  p_j(y) \sigma^j_y \otimes |y\rangle\langle y|    \right),\end{align*}
where $\rho^j_x$ and  $\sigma^j_y$ are states of registers $\bF^j_L$ and $\bF^j_R$ respectively. Then the marginal distribution of $XY$ is $p(x,y)=\sum_j q(j)p_j(x) p_j(y)$, and we have
\begin{align*}
\tau_{\bF Y} & = \bigoplus_j  \sum_x q(j) \left(  \sum_x  p_j(x)  \rho^j_x    \right)\otimes \left(   \sum_y  p_j(y)  \sigma^j_y \otimes |y\rangle\langle y|    \right)\\
& = \sum_y p(y) \bigg(  \bigoplus_j  \frac{1}{p(y)} \sum_x q(j)p_j(x) p_j(y)   \rho_x^j\otimes \sigma_y^j \bigg)\otimes \ket{y}\bra{y},
\end{align*}
where $p(y) =  \sum_{j, x} q(j)p_j(x)p_j(y)$ is the marginal distribution of $Y$. Then conditioned on $Y=y$ the state of $\bF$ is
\begin{align}\label{eq:rho-y}
\bigoplus_j  \frac{1}{p(y)}  \sum_x    q(j) p_j(x) p_j(y) \rho_x^j\otimes \sigma_y^j.
\end{align}
To find a purification of this state let $|\psi^j_x\rangle_{\bE^j_L\bF^j_L}$ and $|\phi^j_y\rangle_{\bE^j_R\bF^j_R}$ be purifications of $\rho^j_x$ and $\sigma_y^j$ respectively. Then\small
$$ \sum_{j, x}     \sqrt{ \frac{q(j)p_j(x) p_j(y) }{p(y)}} \, | j \rangle_{\bJ} | x\rangle_{\mathbf X'} |\psi^j_x\rangle_{\bE^j_L\bF^j_L}|\phi^j_y\rangle_{\bE^j_R\bF^j_R} $$\normalsize
is a purification of~\eqref{eq:rho-y} where the register which purifies $\bF$ is $\bE=\bJ\mathbf X'\left(\bigoplus_j \bE^j_L \bE^j_R \right)$. Note that all purifications of~\eqref{eq:rho-y} are equivalent to the above purification up to a unitary, and $\bE$ contains a copy of $X$ as a subsystem with distribution $\sum_j q(j)p_j(x)p_j(y)/p(y) =p(x|y)$. We are done.

\end{proof}

\begin{lemma}\label{lemma:lemma1}
Suppose that $\textbf{F}-X-Y$. This means that the joint state of $\bF XY$ has the form
$$\rho_{\bF XY} = \sum_{x,y} p(x,y) \rho_x \otimes \ket x\bra x\otimes \ket y\bra y.$$
Let $x^ny^n$ be jointly typical. Then for sufficiently large $n$ we have
$$\tr\left(\Pi_{\delta}^{y^n}\rho_{x^n}\right)\geq 1-\epsilon.$$
\end{lemma}

\begin{proof}
Let $\mathcal{Y}=\{1, 2, \dots, k\}$. Without loss of generality we may assume that $y^n$ has the form
$$y^n=\underbrace{1\dots 1}_{\ell_1}\underbrace{2\dots 2}_{\ell_2}\underbrace{3\dots 3}_{\ell_3}\cdots \underbrace{k\dots k}_{\ell_k},$$
where $\sum_{y=1}^k \ell_y=n$. Since $y^n$ is typical for all $y\in \{1, \dots, k\}$ we have
$$\left|\frac{\ell_y}{n}-p(y)\right|< \delta,$$
which implies that $\ell_y$ is sufficiently large when $n$ is sufficiently large (and $p(y)>0$).

By the definition of the conditional typical projection we have
$$\Pi_{\delta}^{y^n}=\Pi_{\sigma_1,\delta}^{\ell_1}\otimes\cdots\otimes\Pi_{\sigma_k,\delta}^{\ell_k},$$
where $\Pi_{\sigma_y,\delta}^{\ell_y}$ is the typical projection with respect to $\sigma_y=\sum_{x}p(x|y)\rho_{x}$.

With abuse of notation we may write
$$\rho_{x^n}=\rho_{x_1}\otimes\rho_{x_2}\otimes\cdots \otimes\rho_{x_n}=\rho_{x^{\ell_1}}\otimes\rho_{x^{\ell_2}}\otimes\cdots \otimes\rho_{x^{\ell_k}}.$$
$x^{\ell_y}$ is $\delta$-typical with respect to $p(x|y)$ since $x^ny^n$ is jointly typical. Thus for sufficiently large $\ell_y$, we have
$$\tr\left(\Pi_{\sigma_y,\delta}^{\ell_y}\rho_{x^{\ell_y}}\right)\geq 1-\frac{\epsilon}{k}.$$
As a result,
\begin{align*}\tr\left(\Pi_{\delta}^{y^n}\rho_{x^{n}}\right)&=\prod_{y=1}^k \tr\left(\Pi_{\sigma_y,\delta}^{\ell_y}\rho_{x^{\ell_y}}\right) \geq \left(1-\frac{\epsilon}{k}\right)^k \geq 1-\epsilon.\end{align*}
\end{proof}


\section{Proof of Theorem \ref{thm:q-region}}\label{AppendixProofCorrelation}

\subsection{Proof of $\mathcal{R}^q\subseteq \bigcap_{\epsilon>0}\mS_\epsilon$}
The proof follows from similar steps as in the classical case.
For every $\epsilon>0$ we show that $\mR^q\subseteq \mS_{\epsilon}$. Let $R\in \mR^q$. Then by definition there exists an $(n, \epsilon, R)$ code for a sufficiently large $n$ such that $H(M) = nR$. Let $\bQ$ denote Bob's part of the shared entangled state \emph{after Alice's measurement}. We have
\begin{align}
nR&= H(M)\\&\geq I(X^n;M|Y^n\textbf{Q})\n &=I(X^n;M\textbf{Q}|Y^n) - I(X^n: \textbf{Q} |Y^n)\n
& = I(X^n;M\textbf{Q}|Y^n)\label{eqn:P1}\\
&=\sum_{i=1}^nI(X_{i};M\textbf{Q}|X_{1:i-1}Y_{i}Y_{1:i-1}Y_{i+1:n})\n&= \sum_{i=1}^nI(X_{i};M\textbf{Q}X_{1:i-1}Y_{1:i-1}Y_{i+1:n}|Y_{i}),\label{eqn:P2}
\end{align}
where \eqref{eqn:P1} follows from the no-signaling principle and \eqref{eqn:P2} follows from the fact that $X^n, Y^n$ are drawn independently. Let $U$ be a random variable uniformly distributed over $\{1,2,\cdots,n\}$ and independent of all previously defined registers. Let $X=X_{U}$, $Y=Y_{U}$, $A=A_{U}$, $B=B_{U}$, and
$$\textbf{F}=(M,\textbf{Q},X_{1:U-1},Y_{1:U-1}Y_{U+1:n},U).$$
Thus using equation \eqref{eqn:P2} we can write
\begin{align}R&\geq \frac{1}{n}\sum_{i=1}^nI(X_{i};M\textbf{Q}X_{1:i-1}Y_{1:i-1}Y_{i+1:n}|Y_{i})\nonumber
\\&=I(X;\textbf{F}|YU)\nonumber
\\&=I(X;\textbf{F}U|Y)-I(X;U|Y)\label{eqn:rvadded1}
\\&=I(X;\textbf{F}|Y),\nonumber
\end{align}
where in \eqref{eqn:rvadded1}, $I(X;U|Y)=0$ holds because $U$ is independent of $(X,Y)$. This is because $X^n$ and $Y^n$ are i.i.d. and hence the joint distribution of $(X=X_{U}$, $Y=Y_{U})$ conditioned on $U=u$ is the same as that without conditioning on $U=u$. Therefore if we show that $A, B, X, Y, \bF$ satisfy the conditions given by~\eqref{eq:f-q}, we are done.

By the definition of $X=X_{U}$, $Y=Y_{U}$, $A=A_{U}$, $B=B_{U}$, the probability distribution over $A, B, X, Y$ is $\tilde{p}(a,b,x,y)=\frac{1}{n}\sum_{i=1}^n\tilde{p}(a_i,b_i,x_i,y_i)$. Therefore,
\begin{align*}
\tv{\tilde{p}(a,b,x,y)-p(a,b,x,y)} &  =\tv{\frac{1}{n}\sum_{i=1}^n\tilde{p}(a_i,b_i,x_i,y_i)-p(a,b,x,y)} \\
& \leq \frac{1}{n}\sum_{i=1}^n\tv{\tilde{p}(a_i,b_i,x_i,y_i)-p(a,b,x,y)} \\
& \leq \tv{\tilde{p}(a^n,b^n,x^n,y^n)-\prod_{i=1}^n p(a_i,b_i,x_i,y_i)}\\
& \leq \epsilon.
\end{align*}

Next to show that $\textbf{F}-X-Y$, observe that
\begin{align*}
I(\textbf{F};Y|X) &=I(M\textbf{Q}X_{1:U-1}Y_{1:U-1}Y_{U+1:n}U;Y_{U}|X_{U})\\
&=\frac{1}{n}\sum_{i=1}^nI(M\textbf{Q}X_{1:i-1}Y_{\neg i};Y_{i}|X_{i})\\
&\leq\frac{1}{n}\sum_{i=1}^nI(M\textbf{Q}X_{\neg i}Y_{\neg i};Y_{i}|X_{i})\\
&=\frac{1}{n}\sum_{i=1}^nI(M\textbf{Q};Y_{i}|X_{i}X_{\neg i}Y_{\neg i})\\
&\leq\frac{1}{n}\sum_{i=1}^nI(M\textbf{Q};Y^n|X^n)\\
&=I(M\textbf{Q};Y^n|X^n)\\
&=0,
\end{align*}
where $X_{\neg i}$ and $Y_{\neg i}$ denotes $X_{1:i-1}X_{i+1:n}$ and $Y_{1:i-1}Y_{i+1:n}$ respectively; the last step follows from the Markov chain condition $Y^n-X^n-M\textbf{Q}$, i.e., $M$ and $\textbf{Q}$ are generated from $X^n$ independent of $Y^n$.

Next to show that $A-\textbf{F}X-Y$, note that
\begin{align*}
I(A;Y|X\textbf{F})&=I(A\textbf{F};Y|X)
\\&=I(A_{U}M\textbf{Q},X_{1:U-1}Y_{1:U-1}Y_{U+1:n}U;Y_{U}|X_{U})
\\&=\frac{1}{n}\sum_{i=1}^nI(A_{i}M\textbf{Q}X_{1:i-1}Y_{\neg i};Y_{i}|X_{i})
\\&\leq\frac{1}{n}\sum_{i=1}^nI(A^nM\textbf{Q}X_{\neg i}Y_{\neg i};Y_{i}|X_{i})
\\&=\frac{1}{n}\sum_{i=1}^nI(A^nM\textbf{Q};Y_{i}|X_{i}Y_{\neg i}X_{\neg i})
\\&\leq\frac{1}{n}\sum_{i=1}^nI(A^nM\textbf{Q};Y^n|X^n)
\\&=I(A^nM\textbf{Q};Y^n|X^n)
\\&=0,
\end{align*}
where again the last step follows from the Markov chain condition $Y^n-X^n-A^nM\textbf{Q}$, i.e., $A^nM\textbf{Q}$ are generated from $X^n$ independent of $Y^n$.

Lastly to show that $\Psi(\textbf{F},Y)=(B,Y)$ for some measurement $\Psi$ on $\textbf{F},Y$, note that $\textbf{F},Y$ includes $M, \textbf{Q},Y_{1:U-1}Y_{U+1:n}, U, Y_{U}$ meaning that it contains $M, \textbf{Q},Y^n,U$. Thus, we can use the measurement on $M, \bQ, Y^n$ in the code which gives $B^n$, and depending on the value of $U$ construct $B=B_U$ out of $\bF , Y$.\\


\subsection{Proof of $\mathcal{S}_0\subseteq\mathcal{R}^q$}

To prove the achievability it would be helpful to start with a simpler problem, namely \emph{remote state preparation with classical side information}. Remote state preparation has appeared in the literature~\cite{RSP-L, RSP-D, RSP-B, RSP-S, RSP-H}, but here we need a generalization of such protocols that include (classical) side information.

The setup of this problem is as follows. Let $X, Y$ be two random variables with joint distribution $p(x, y)$, and let $\bF$ be a quantum register such that
$$\bF- X-Y.$$
This means that the joint state of  $XY\bF$ is of the form
$$\rho_{XY\bF} = \sum_{x, y} p(x, y)\ket x\bra x\otimes \ket y\bra y\otimes \rho^{\bF}_x.$$
Alice and Bob receive i.i.d.\ repetitions of $X,Y$, and their goal is to prepare i.i.d. repetitions of $\bF$ at Bob's side. The question is how much classical communication from Alice to Bob is required if they are provided with arbitrary amount of {shared entanglement}. Formally speaking, we can define this problem as follows:

\begin{definition}
Alice and Bob receive $n$  i.i.d.\ repetitions of $X,Y$, i.e., $X^n$  and $Y^n$ respectively. They are also provided with an entangled state on registers $\mathbf{A}, \mathbf{B}$ (that is independent of $X^nY^n$). An $(n, \epsilon, R)$ remote state preparation code consists of 
\begin{itemize}
\item \emph{An Encoder:} A quantum channel $\mathcal{E}^{\mathbf{A}X^n\rightarrow C}$ where $C$ is a classical random variable taking values in $\{1,2,\cdots, 2^{n(R+\epsilon)}\}$. $C$ is the classical message sent from Alice to Bob;
\item  \emph{A Decoder:} A quantum channel $\mathcal{D}^{\mathbf{B}Y^nC\rightarrow Y^n\tilde{\bF}^n}$.
\end{itemize}
The total variation distance of an $(n, \epsilon, R)$ code is defined as follows. Let $\sigma^{X^nY^n \tilde{\bF}^n}$ be the induced state by the code, i.e., the result of applying $\mathcal{E}$ on $(\mathbf{A}, X^n)$ and then applying $\mathcal{D}$ on $(\mathbf{B}, Y^n, C)$ to produce $(Y^n,\tilde{\bF}^n)$. Then the  total variation distance of the code is
\be
\label{eq:t}
\tv{\sigma^{X^n Y^n \tilde{\bF}^n}-(\rho_{XY\bF})^{\otimes n}}.
\ee
\end{definition}

Given $\rho_{XY\bF}$, a rate $R$ is said to be achievable if there exists a sequence of $(n, R, \epsilon_n)$ remote state preparation codes for $n\in\mathbb{N}$ such that 
$$\lim_{n\rightarrow\infty}\epsilon_n=0,$$
and further the  total variation distance of the code converges to zero as $n$ converges to infinity.

\begin{theorem}\label{thm:preparation}
\emph{(Remote state preparation with classical side information)} The minimum achievable rate of one-way (classical) communication for remote state preparation with classical side information and arbitrary amount of shared entanglement is $I(X; \bF| Y)$.
\end{theorem}

$\mS_0\subseteq \mR^q$ is a simple consequence of this theorem. Let $\bF$ be a quantum register satisfying conditions~\eqref{eq:f-q} for $\epsilon =0$. By the above theorem Alice can prepare an approximation of $\bF^n$ at Bob's side with almost $nI(X; \bF| Y)$ bits of one-way communication. In the remote state preparation protocol Alice has an approximate purification of $\bF^n$ in hand (see the details of the proof below). Thus using $A-\bF X-Y$ and based on Lemma~\ref{lem:purification} she can generate an approximation of $A^n$. On the other hand,
since by~\eqref{eq:f-q} there is a measurement on $(\bF, Y)$ which gives $B$, Bob can generate an approximation of $B^n$ after receiving Alice's message. The details are straightforward, so we only need to prove Theorem~\ref{thm:preparation}.\\

\begin{proof}
We start by showing that at least $I(X; \bF| Y)$ bits of communication per copy is required. Suppose that for every $\epsilon>0$ and sufficiently large $n$ there is a protocol with $nR$ bits of communication in which Bob can prepare $\widetilde{\bF}^n$ such that the trace distance between the state of $(X^n, Y^n, \widetilde{\bF}^n)$ and $(X^n, Y^n, \bF^n)$ is at most $\epsilon$. Then by Fannes inequality we have
\begin{align}\label{eq:com1}
I(X^n ; \bF^n| Y^n)  -n \epsilon\log d  - \eta(\epsilon)/\ln 2 \leq I(X^n; \widetilde{\bF}^n| Y^n), 
\end{align}
where $d=(\dim \bF) |\mx|\cdot |\my|$ and $\eta(\epsilon) = -\epsilon \ln \epsilon$. Let $M$ be the message from Alice to Bob and $\bQ$ be Bob's part of the shared entangled state after receiving $M$. Then by the date processing inequality we have
\begin{align*}
I(X^n; \widetilde{\bF}^n| Y^n) & \leq I(X^n ; M\bQ | Y^n) \\
& = I(X^n ; \bQ | Y^n ) + I(X^n ; M| Y^n \bQ) \\
& = I(X^n ; M | Y^n\bQ) \\
& \leq H(M) \\
& \leq nR,
\end{align*}
where in the third line we use the no-signaling principle. Combining the above inequality with~\eqref{eq:com1} gives the desired result.

We now discus the achievability protocol. Our protocol is based on the \emph{column method} for remote
state preparation~\cite{RSP-D, RSP-B}.
For a sufficiently large $n$ Alice and Bob share $2^{n(I(\textbf{F};X)+\delta''+c\delta+\alpha)}$ copies of $|\psi\rangle_{\bE'\bF'}$, where $\ket {\psi}_{\bE'\bF'}$ is proportional to
$$\ket{\psi}_{\bE'\bF'}\sim I\otimes{\Pi_{\delta}^n}|\Phi\rangle_{\bE'\bF'}.$$
Here $\ket{\Phi}_{\bE'\bF'}$ is the maximally entangled state,
and $\Pi_{\delta}^n= \Pi_{\rho, \delta}^n$ is the typical projection of
$\rho=\sum_x p(x)\rho_x$. We let
$$\tau^n_{\delta} = \tr_{\bE'} (\ket{\psi}\bra{\psi}_{\bE'\bF'})= \frac{1}{\tr \Pi_{\delta}^{n}}\Pi_{\delta}^{n}.$$
Thus Alice holds copies of $\bE'$ and Bob holds copies of $\bF'$. They put these copies in groups of size $2^{n(I(\textbf{F};Y)-\delta''-c\delta-\alpha)}$. Thus the number of groups is equal to
$$\frac{2^{n(I(\textbf{F};X)+\delta''+c\delta+\alpha)}}{2^{n(I(\textbf{F};Y)-\delta''-c\delta-\alpha)}}=2^{n(I(\textbf{F};X)-I(\textbf{F};Y)+2\delta''+2c\delta+2\alpha)}.$$

Alice and Bob respectively receive $x^n$ and $y^n$. For sufficiently large $n$, with probability at least $1-\epsilon$, $x^ny^n$ is jointly typical. Alice measures her side of $|\psi\rangle_{\bE'\bF'}$ for all copies using the measurement $\{Q_{\delta}^{x^n}, \sqrt{I-(Q_{\delta}^{x^n})^2}\}$ where
$$Q_{\delta}^{x^n}=\sqrt{2^{n(H(\textbf{F}|X)-\delta'')}}\left(\sqrt{\Pi_{\delta}^{x^n}\rho_{x^n}\Pi_{\delta}^{x^n}}\right)^T.$$
Note that by the properties of typical projections mentioned in Appendix~\ref{app:useful-lem}, $\{Q_{\delta}^{x^n}, \sqrt{I-(Q_{\delta}^{x^n})^2}\}$ is indeed a valid measurement.

The following lemma is proved in Appendix~\ref{app:lemmas}.

\begin{lemma}\label{lemma:lemma2}
If we measure $|\psi\rangle_{\bE'\bF'}$ by the measurement $\{Q_{\delta}^{x^n}, \sqrt{I-(Q_{\delta}^{x^n})^2}\}$ acting on subsystem $\bE'$, then the outcome would be $Q_{\delta}^{x^n}$ with probability at least
$$2^{-n(I(X;\textbf{F})+\delta''+c\delta)} (1-4\sqrt{\epsilon}),$$
and in this case $\bF'$ collapses to some $\rho_{x^n}''$ where
$$\|\rho_{x^n}-\rho_{x^n}''\|_1\leq 6\sqrt[4]{\epsilon}.$$
As a result, if this measurement is applied on
$2^{n(I(\textbf{F};X)+\delta''+c\delta+\alpha)}$ copies of $|\psi\rangle_{\bE'\bF'}$, with probability at least $1-e^{-(1-4\sqrt{\epsilon})2^{\alpha n}}$ one of the outcomes is $Q_{\delta}^{x^n}$.
\end{lemma}

This lemma states that with probability at least $1-e^{-(1-4\sqrt{\epsilon})2^{\alpha n}}$ there exists an index  $i$ (if there are more than one pick one randomly) such that the outcome of the $i$-th measurement is $Q_{\delta}^{x^n}$. Then Bob's side of the $i$-th copy of shared entangled states collapses to $\rho_{x^n}''$. Alice sends Bob the index of the group to which $i$ belongs. She needs $n(I(\textbf{F};X)-I(\textbf{F};Y)+2\delta''+2c\delta+2\alpha)$ bits of communication to send this index.

Now Bob applies the measurement $\{\Pi_{\delta}^{y^n}, I-\Pi_{\delta}^{y^n}\}$ on all subsystems in the group to which $i$ belongs. These measurements are indeed measurements on $\tau_{\delta}^n=\tr_{\bE'}(|\psi\rangle\langle \psi |_{\bE'\bF'})$.

\begin{lemma}\label{lemma:lemma5} If we apply the measurement $\{ \Pi_{\delta}^{y^n}, I-\Pi_{\delta}^{y^n} \}$ on $2^{n(I(\textbf{F};Y)-\delta''-c\delta-\alpha)}$ copies of $\tau_{\delta}^n$, the probability of obtaining more than one
$\Pi^{y^n}_\delta$ is at most $\frac{2^{-\alpha n+1}}{(1-\epsilon)}$.
\end{lemma}

The above lemma is proved in Appendix~\ref{app:lemmas}. Thus among Bob's measurements with high probability there is at most one outcome $\Pi^{y^n}_\delta$. On the other hand for the $i$-th subsystem we have
$$\tr\left(\Pi_{\delta}^{y^n}\rho_{x^n}''\right)\geq \tr\left(\Pi_{\delta}^{y^n}\rho_{x^n}\right)-6\epsilon^{\frac{1}{4}}\geq 1-\epsilon-6\epsilon^{\frac{1}{4}},$$
where here we use Lemma~\ref{lemma:lemma1}. Therefore, this measurement helps Bob to distinguish the index $i$.
In fact by the gentle measurement lemma with high probability the measurement on the $i$-th subsystem results in a state $\rho_{x^n}'''$ such that $$\|\rho_{x^n}-\rho_{x^n}'''\|_1\leq
\|\rho_{x^n}-\rho_{x^n}''\|_1 + \|\rho_{x^n}''-\rho_{x^n}'''\|_1$$$$~~\leq 6\epsilon^{\frac{1}{4}}+2\sqrt{\epsilon+6\epsilon^{\frac{1}{4}}}.$$

The probability of error of the protocol can be analyzed as follows. In the first part of the protocol where Alice measures, with probability at most $e^{-(1-4\sqrt \epsilon)2^{\alpha n}}$ Alice obtains no $Q_{\delta}^{x^n}$ as the outcome which results in error. Otherwise there exists $i$ such that Bob's part of the $i$-th system is within $6\epsilon^{\frac 1 4}$ of a desired state ($\rho_{x^n}$). Then in the second step where Bob measures the probability of detect some $j\neq i$ is at most $\frac{2^{-\alpha n+1}}{(1-\epsilon)}$, and the probability of correctly detecting is at least $1-\epsilon -6\epsilon^{\frac 1 4}$. Adding all these error terms, the probability of error of the whole protocol is at most
$$e^{-2^{n\alpha}(1-4\sqrt{\epsilon})}+\frac{2^{-\alpha n+1}}{1-\epsilon}+\epsilon+6\epsilon^{\frac{1}{4}},$$
which tends to zero as $\epsilon\rightarrow 0$ and $n\rightarrow \infty$. Moreover, the number of communicated bits is equal to
$$n(I(\textbf{F};X)-I(\textbf{F};Y)+2\delta''+2c\delta+2\alpha) = n(I(X; \bF |Y)  +2\delta''+2c\delta+2\alpha ).$$

\end{proof}

\subsection{Proof of Lemmas \ref{lemma:lemma2} and \ref{lemma:lemma5}} \label{app:lemmas}

\noindent\textbf{Proof of Lemma~\ref{lemma:lemma2}:} Observe that
\begin{align*}\left(Q_{\delta}^{x^n}\right)\left(Q_{\delta}^{x^n}\right)^\dag&=\left(Q_{\delta}^{x^n}\right)^2\\&=2^{n(H(\textbf{F}|X)-\delta'')}\left(\Pi_{\delta}^{x^n}\rho_{x^n}\Pi_{\delta}^{x^n}\right)^T\\&\leq \left(\Pi_{\delta}^{x^n}\right)^T\\&\leq I.\end{align*}
Thus $\{Q_{\delta}^{x^n}, \sqrt{I-(Q_{\delta}^{x^n})^2}\}$ defines a measurement. Moreover,
\begin{align}
\tr_{\bE'} \left[\left(Q_{\delta}^{x^n}\otimes I_{\bF'}\right)|\psi\rangle\langle\psi|_{\bE'\bF'}\left(Q_{\delta}^{x^n}\otimes I_{\bF'}\right)\right]\nonumber &=\sqrt{\tau_{\delta}^n}\left(\left(Q_{\delta}^{x^n}\right)^T\right)^2\sqrt{\tau_{\delta}^n}\nonumber\\
&=\frac{2^{n(H(\textbf{F}|X)-\delta'')}}{\tr \Pi_{\delta}^{n}}\Pi_{\delta}^{n}\Pi_{\delta}^{x^n}\rho_{x^n}\Pi_{\delta}^{x^n}\Pi_{\delta}^{n}.
\label{eqn:S1}
\end{align}
Now the results follows from two applications of the gentle measurement lemma (Lemma~\ref{lem:GML}). Define
\[\rho'_{x^n}=\frac{\Pi_{\delta}^{x^n}\rho_{x^n}\Pi_{\delta}^{x^n}}{\tr\left(\Pi_{\delta}^{x^n}\rho_{x^n}\Pi_{\delta}^{x^n}\right)}, \quad
\rho''_{x^n}=\frac{\Pi_{\delta}^{n}\rho'_{x^n}\Pi_{\delta}^{n}}{\tr\left(\Pi_{\delta}^{n}\rho'_{x^n}\Pi_{\delta}^{n}\right)}.
\]
We have $\tr\left(\Pi_{\delta}^{x^n}\rho_{x^n}\right)\geq 1-\epsilon$. Thus $\|\rho_{x^n}-\rho'_{x^n}\|_1\leq 2\sqrt{\epsilon}$, and
$$\tr\left(\rho'_{x^n}\Pi_{\delta}^{n}\right)\geq \tr\left(\rho_{x^n}\Pi_{\delta}^{n}\right)-2\sqrt{\epsilon}\geq 1-\epsilon-2\sqrt{\epsilon}.$$
Therefore,  $\|\rho_{x^n}'-\rho_{x^n}''\|_1\leq 2\sqrt{\epsilon+2\sqrt{\epsilon}}$, and by the triangle inequality
$$\|\rho_{x^n}-\rho_{x^n}''\|_1\leq 2\left(\sqrt{\epsilon}+\sqrt{\epsilon+2\sqrt{\epsilon}}\right)<6\sqrt[4]{\epsilon}.$$

We can work out equation \eqref{eqn:S1} as follows:
\begin{align*}
&\frac{2^{n(H(\textbf{F}|X)-\delta'')}}{\tr \Pi_{\delta}^{n}}\tr\left(\Pi_{\delta}^{x^n}\rho_{x^n}\right)\Pi_{\delta}^{n}\rho'_{x^n}\Pi_{\delta}^{n} =\frac{2^{n(H(\textbf{F}|X)-\delta'')}}{\tr \Pi_{\delta}^{n}}\tr\left(\Pi_{\delta}^{x^n}\rho_{x^n}\right)\tr\left(\Pi_{\delta}^{n}\rho'_{x^n}\right)\rho''_{x^n}.
\end{align*}
This means that after the measurement, subsystem $\bF'$ collapses to $\rho''_{x^n}$ with probability
\begin{align*}
\frac{2^{n(H(\textbf{F}|X)-\delta'')}}{\tr \Pi_{\delta}^{n}}\tr\left(\Pi_{\delta}^{x^n}\rho_{x^n}\right) \tr\left(\Pi_{\delta}^{n}\rho'_{x^n}\right)
&\geq
2^{n(H(\textbf{F}|X)-\delta''-H(\textbf{F})-c\delta)}(1-\epsilon)(1-(\epsilon+2\sqrt{\epsilon}))
\\&=2^{-n(I(\textbf{F};X)+\delta''+c\delta)}(1-\epsilon)(1-(\epsilon+2\sqrt{\epsilon}))
\\&>2^{-n(I(\textbf{F};X)+\delta''+c\delta)}(1-4\sqrt{\epsilon}).
\end{align*}\\

\noindent\textbf{Proof of Lemma~\ref{lemma:lemma5}:}
First note that
\begin{align*}\tr\left(\tau_{\delta}^n\Pi_{\delta}^{y^n}\right)&\leq \frac{\tr \Pi_{\delta}^{y^n}}{\tr \Pi_{\delta}^{n}}\\&\leq \frac{2^{n(H(\textbf{F}|Y)+\delta'')}}{\tr \Pi_{\delta}^{n}}\\&\leq \frac{1}{1-\epsilon}2^{-n(I(\textbf{F};Y)-\delta''-c\delta)}.\end{align*}
Let $m=2^{n(I(\textbf{F};Y)-\delta''-c\delta-\alpha)}$ and $t=\frac{1}{1-\epsilon}2^{-n(I(\textbf{F};Y)-\delta''-c\delta)}$. Thus the probability of obtaining one or no $\Pi^{y^n}_\delta$ as the measurement outcome, is at least the probability of obtaining no $\Pi^{y^n}_\delta$ as the outcome. The latter probability is greater than or equal to
\begin{align*}(1-t)^m= e^{m\ln(1-t)}&\geq e^{-mt/(1-t)}\geq e^{-2mt} \\& =e^{-2^{-n\alpha+1}/(1-\epsilon)}\\&\geq 1- \frac{2^{-\alpha n+1}}{1-\epsilon}.\end{align*}
We are done.


\section{Ad hoc bounds for simulating CHSH-type correlations}\label{sec:CHSHExample}
Consider the problem of simulating the following correlation for a given $\epsilon\geq 0$. Let $x, y$ as well as $a, b$ be binary random variables, and $p(a, b| x,y)$ be equal to $\frac{1+\epsilon}{4}$ when
 $a\oplus b=xy$ holds. This correlation corresponds to a winning strategy for the CHSH game~\cite{CHSH} with probability $p=\frac{1+\epsilon}{2}$.
Here we consider the uniform distribution on inputs ($p(x,y)=p(x)p(y)=\frac{1}{4}$), and study the problem of simulation of this correlation in both classical and quantum settings.

Should we allow for arbitrary preshared randomness, the communication cost would be given by the expression discussed in Theorem~\ref{thm:yassaee}, i.e., the minimum of $I(X;U\vert Y)$ over all classical random variables $U$ determined by $p(u\vert x,y,a,b)$ such that the joint distribution $p(u,a,b,x,y)$ factorizes as
$$p(u,a,b,x,y)=p(x,y)p(u|x)p(a|u,x)p(b|u,y).$$

Independence of $X$ and $Y$ implies that $I(X;U\vert Y)=I(X;U)$. Moreover, $U$ can be taken to be a binary random variable using the Fenchel-Eggleston extension of the Carath\'{e}odory theorem (see \cite[Appendix C]{AbbasYoungHan}). Then computing the optimal rate for every $\epsilon$ is a straightforward optimization problem. The solid line curve of Figure~\ref{fig:chsh} gives the one-way communication cost of winning the CHSH game with probability $p=\frac{1+\epsilon}{2}$.

Pironio \cite{Pironio} shows that the \emph{average} amount of communication required to simulated CHSH-type correlations in the \emph{one-shot case} is equal to $2\epsilon-1$, which is a linear upper bound on our curve that is tight at the end points. Moreover, Roland and Szegedy~\cite{Roland09} prove that the rate of one-way communication required to simulate parallel repetitions of these correlations in the communication complexity setting (without allowing any error) is equal to $1-h(\epsilon)$, where $h(\cdot)$ is the binary entropy function. By definitions the solid line curve of Figure~\ref{fig:chsh} is a lower bound on $1-h(\epsilon)$. However interestingly this curve obtained from numerical simulations is very close to $1-h(\epsilon)$, suggesting that probably the one-way communication cost of simulating CHSH-type correlations is equal to $1-h(\epsilon)$ even in the information theoretic setting.

From the plot of Figure~\ref{fig:chsh} we observe that at $p=\frac{1}{2}+\frac{1}{2\sqrt{2}}$ we get a positive rate while in the presence of entanglement $p=\frac{1}{2}+\frac{1}{2\sqrt{2}}$ can be achieve with no communication. This means that, unlike the problem of point-to-point channel capacity, entanglement does help in the problem of simulation of correlations, even in the information theoretic setting.

\begin{figure*}
\begin{center}
\includegraphics[width=6in]{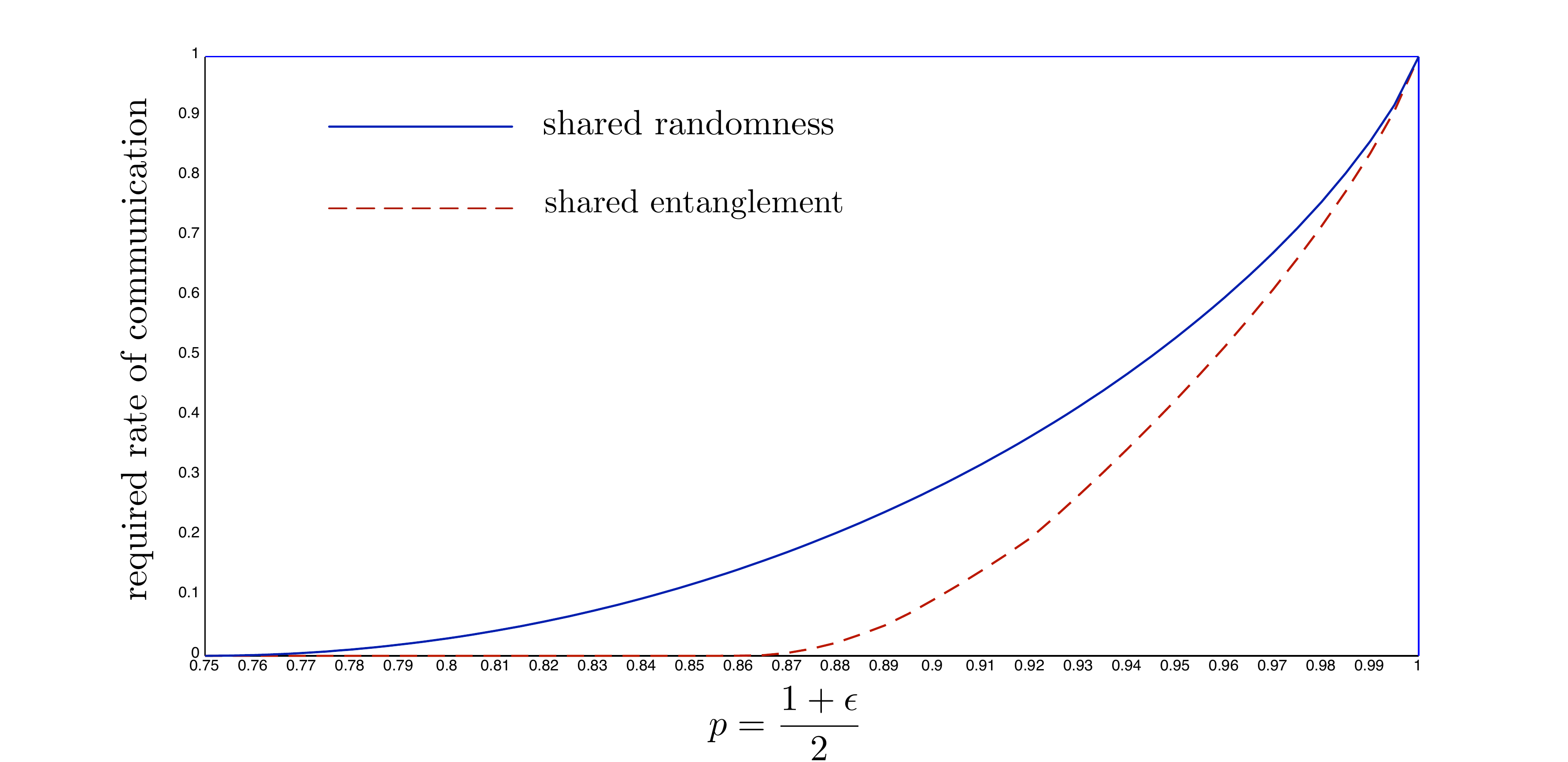}
\caption{(Solid line curve) The one-way communication cost of simulating the CHSH-type correlation with bias $\epsilon$ assuming preshared randomness. The horizontal axes corresponds to parameter $p=\frac{1+\epsilon}{2}$ ($0.75\leq p\leq 1$). (Dashed line curve) A lower bound on the entanglement-assisted one-way communication cost of simulating the CHSH-type correlation with bias $\epsilon=2p-1$. This lower bound is an implication of Information Causality.
}\label{fig:chsh}
\end{center}
\end{figure*}

Should we allow for arbitrary shared entanglement, we have lower and upper bounds on the communication cost by Theorem~\ref{thm:q-region}. Unfortunately both the lower and the upper bounds are non-computable since we have no bound on the dimension of the auxiliary register $\textbf{F}$. Therefore, to find a computable bound we use an ad hoc technique based on the principle of Information Causality~\cite{IC}.

Information Causality is based on the following communication scenario. Alice receives a binary string $a_1\dots a_N$ chosen according to the uniform distribution, and Bob receives a random $b\in \{1, \dots, N\}$. Alice sends a message $m$ to Bob whose goal is to find $a_b$. Letting $g_b$ be Bob's guess of $a_b$ (after communication), Information Causality states that
\begin{align}\label{eq:ic-inequality}
H(M) \geq \sum_{i=1}^N   I(A_i; G_i | B=i ).
\end{align}
In fact the above inequality holds in any physical theory, including the quantum theory, that admits a mutual information function satisfying certain natural properties.

It is shown in~\cite{IC} that Alice and Bob by sharing $k=2^n-1$ no-signaling boxes with CHSH-type correlations with bias $\epsilon$ and sending only one bit from Alice to Bob, can play the above game for $N=2^n$ in such a way that the right hand side of~\eqref{eq:ic-inequality} be equal to $2^n\left( 1- h(\frac{1+\epsilon^n}{2})  \right)$, where $h(\cdot)$ denotes the binary entropy function.  We now would like to simulate this scheme by two new parties, say Alice$^\prime$ and Bob$^\prime$, who instead of non-local boxes, have shared entanglement as their resources at the outset.

Let $R^q$ be the entanglement-assisted communication cost of simulating the non-local box with bias $\epsilon$. Alice$^\prime$ and Bob$^\prime$ can simulate the scheme of Alice and Bob by first sending $kR^q$ bits from Alice$^\prime$ to Bob$^\prime$ to simulate the $k$ boxes, and then one bit to simulate the message that was passed from Alice to Bob. This enables Bob$^\prime$ to faithfully simulate $g_i$. Now since Alice$^\prime$ and Bob$^\prime$ play the game in a quantum world for which Information Causality holds, we may use inequality~\eqref{eq:ic-inequality}. The right hand side of~\eqref{eq:ic-inequality} is equal to $2^n\left( 1- h(\frac{1+\epsilon^n}{2})  \right)$ and the left hand side, namely the number of communicated bits is $kR^q+1$. Therefore,
$$(2^n-1)R^q+1\geq 2^n\left(1-h\left(\frac{1+\epsilon^n}{2}\right)\right),$$
which implies
$$R^q\geq \frac{2^n}{2^n-1}\left(1-h\left(\frac{1+\epsilon^n}{2}\right)\right)-\frac{1}{2^n-1}.$$
Computing this lower bound for all $n$ and taking the optimal one for every $\epsilon$, we obtain the dashed line curve of Figure~\ref{fig:chsh}.
We see that the lower bound is equal to one at $p=1$, thus it has to be tight at this point. By~\cite{IC}, the above lower bound (for $n$ converging to infinity) would also be tight at  the other end point $p\leq \frac{1}{2} + \frac{1}{2\sqrt{2}}$. However, it may be loose in between because firstly we have considered the specific scheme of \cite{IC} for using boxes, and secondly this lower bound holds more generally for any physical theory satisfying properties of mutual information given in~\cite{IC} and not only for quantum physics. Nonetheless, we would like to highlight that the lower bound at $p=1$ is tight in any such physical theory, as shown in the figure.

\section{An $\epsilon$-net of classical channels}\label{app:details}

The following lemmas complete the argument of Section~\ref{sec:new-tools} and follow the same notations.

\begin{lemma}\label{lemmaCard1} There exists a finite set of channels $p(c_j|x), (j=1,2,\cdots,M_{\epsilon})$ such that for every $p(c|x)\in \mathcal{P}$ there exists $j$ such that
\begin{align*}
\big|\big(I(C;Y)-I(C;Z)&\big)-\big(I(C_j;Y)-I(C_j;Z)\big)\big| \leq \epsilon,\qquad \forall q(y,z|x).
\end{align*}
\end{lemma}

\begin{proof}
Define
$$T\left(p(x),p(c|x)\right)=H(C),$$
for all $p(x)$ and $p(c|x)\in \mathcal{P}$. Then $T$ is a continuous function defined on a compact set. Thus $T$ is \emph{uniformly continuous}, i.e., for every $\epsilon>0$ there exists $\delta>0$ such that if $\| (p(x), p(c| x)) - (p(x'), p(c'|x))   \|_{1} \leq \delta$, then
$$|T(p(x), p(c|x)) - T(p(x'), p(c'|x)) |\leq \epsilon/2.$$

On the other hand, the set $\mathcal{P}$ of points $p(c|x)$ is compact. Thus, there exists a $\delta$-net, i.e., there exists a finite set of points $p(c_j|x)$, $j=1, \dots, M_\epsilon$, such that for every $p(c|x)$ there exists $j$ with
$\| p(c| x) - p(c_j|x)   \|_{1} \leq \delta$. This implies that for every $p(x)$,
$$\left|T(p(x), p(c|x)) - T(p(x), p(c_j|x))\right| \leq \epsilon/2.$$
Now note that
\begin{align*}
I(C;Y)-I(C;Z)&=H(C|Z)-H(c|Y)\\
&=\sum_{z}p(z)T(p(x|z),p(c|x)) -\sum_{y} p(y)T(p(x|y),p(c|x)),\end{align*}
is the difference of two convex combinations of evaluations of $T$. 

\end{proof}

\begin{lemma}\label{lemmaCard2} With the notation developed in Section~\ref{sec:new-tools}, we have
\begin{align*}
&\min_{\lambda_j\geq 0:\sum_{j}\lambda_j=1}\,\,\max_{q(y,z|x)}\bigg[I(\textbf{F};Y)-I(\textbf{F};Z)-\sum_{j}\lambda_j\big(I(C_j;Y)-I(C_j;Z)\big)\bigg] \\&
\qquad\quad  = \max_{q(y,z|x)}\,\,\min_{\lambda_j\geq 0:\sum_{j}\lambda_j=1}\bigg[I(\textbf{F};Y)-I(\textbf{F};Z)-\sum_{j}\lambda_j\big(I(C_j;Y)-I(C_j;Z)\big)\bigg]
\end{align*}

\end{lemma}
\begin{proof}
To show the legitimacy of exchanging the maximum and minimum we use Corollary 2 of \cite{GengGohari} with the choice of $T_j(q(y,z|x))=I(\textbf{F};Y)-I(\textbf{F};Z)-I(C_j;Y)+I(C_j;Z)$, and $d=M_{\epsilon}$. To apply this corollary we need to show the convexity of the set
\begin{align*}
&\mathcal{A}=\big\{(a_1,...,a_{M_\epsilon}): a_j\leq I(\textbf{F};Y)-I(\textbf{F};Z)-I(C_j;Y)+I(C_j;Z) \mbox{~for some~}q(y,z|x)\big\}.\end{align*}
Take two arbitrary points in $\mathcal{A}$. We show that their average is in $\mathcal{A}$. Corresponding to these two points are two channels $q(y_1,z_1|x)$ and $q(y_2,z_2|x)$. We construct $q(y_0,z_0|x)$ as follows. Let $U$ be the uniform binary random variable over $\{1,2\}$ independent of all previously defined registers. Let $Y_0=(U,Y_U)$ and $Z_0=(U,Z_U)$. Then it is easy to verify that
\begin{align*}  T_j(q(y_0,z_0&|x))= \frac{1}{2}\big(T_j(q(y_1,z_1|x))+T_j(q(y_1,z_1|x))\big).
\end{align*}
This implies that the average of the two points is in $\mathcal{A}$.
\end{proof}

\section{Proof of Theorem~\ref{thm:QCMain2}}\label{app:carve-d2}

By part (a) of Theorem~\ref{thm:QCMain1} it suffices to prove the first part. Since $|\mathcal{X}|=2$ the distribution of $X$ is determined by $p(X=0)=p$ and $p(X=1)=\bar{p}=1-p$. $I(X; \bF)$ for $p=0$ and $p=1$ is equal to zero and equal to $I(X; C)$ for every $X\rightarrow C$. Therefore we only need to show that there exists a classical channel $X\rightarrow C$ such that
$$\frac{\partial^2}{\partial p^2} I(X; \bF) = \frac{\partial^2}{\partial p^2} I(X; C),$$
at every $p$. On the other hand since $I(X; \bF)$ and $H(\bF)$ differ only at a linear function in terms of $p$, it is sufficient to show that there exists a classical channel $X\rightarrow C$ with
$$\frac{\partial^2}{\partial p^2} H(\mathbf{F}) = \frac{\partial^2}{\partial p^2} H(C).$$

 Let $\overrightarrow{s}=(s_1, s_2, s_3)$ and $\overrightarrow{r}=(r_1, r_2, r_3)$ be the Bloch sphere representations of $\rho_0$ and $\rho_1$ respectively, i.e.,
\begin{align*}
\rho_0 = \frac{1}{2}I + \frac{1}{2}\left( s_1 \sigma_x + s_2 \sigma_y + s_3 \sigma_z   \right),\qquad 
\rho_1 = \frac{1}{2}I + \frac{1}{2}\left( r_1 \sigma_x + r_2 \sigma_y + r_3 \sigma_z   \right),
\end{align*}
where $\sigma_x, \sigma_y, \sigma_z$ are Pauli matrices. If $\overrightarrow{r}= \overrightarrow{s}$ then $\rho_0= \rho_1$ and the existence of $C$ is immediate. Thus we assume $\overrightarrow{r}\neq \overrightarrow{s}$.
The margin of $\mathbf{F}$ is equal to
\begin{align*}&\rho= \rho_p  =   {p}\rho_0 + \bar{p}\rho_1 = \frac{1}{2}I + \frac{1}{2} \bigg(   (\bar{p}r_1+ ps_1)\sigma_x + (\bar{p}r_2+ ps_2)\sigma_y        + (\bar{p}r_3+ ps_3)\sigma_z\bigg),\end{align*}
so the eigenvalues of $\rho$ are
$$\lambda = \lambda_p =  \frac{1 + \|\bar{p}\overrightarrow{r} + p\overrightarrow{s}\|}{ 2},$$
and $1-\lambda$.  Therefore, $H(\mathbf{F}) = h(\lambda)$ where $h(\cdot)$ denotes the binary entropy function. The second derivative of the binary entropy function is computed as
\begin{align}\label{eq:2nd-der-h}
&\frac{\partial^2}{\partial p^2} h(\lambda)  = -\lambda'' (\ln \lambda - \ln (1-\lambda)) -\lambda'^2 \left(\frac{1}{\lambda(1-\lambda)}\right),
\end{align}
where $\lambda' = \frac{\partial}{\partial p} \lambda$ and $\lambda'' = \frac{\partial^2}{\partial p^2} \lambda$, and for simplicity we take the natural logarithm instead of logarithm in base $2$.

Let us define
\begin{align*}Z=Z_p &= \|  \bar{p}\overrightarrow{r} + p\overrightarrow{s} \|^2 =  \|  \overrightarrow{s} - \overrightarrow{r} \|^2p^2 + 2 \langle \overrightarrow{r} \vert \overrightarrow{s}- \overrightarrow{r} \rangle p + \|\overrightarrow{r}\|^2,\end{align*}
and $\Delta = \| \overrightarrow{r} \|^2\cdot \|\overrightarrow{s}- \overrightarrow{r}\|^2 - \langle \overrightarrow{r}\vert \overrightarrow{s} - \overrightarrow{r}\rangle^2$. By Cauchy-Schwarz inequality $\Delta\geq 0$ and $\Delta \leq \|\overrightarrow{s}  - \overrightarrow{r}\|^2$  since $\|s\|, \|r\|\leq 1$. Then $\lambda = \frac{1 + \sqrt{Z}}{2}$ and
\begin{align*}
\lambda' = \frac{Z'}{2\sqrt{Z}}, \quad\quad \lambda'' = \frac{2Z'' Z - Z'^2}{8 Z^{3/2}},
\end{align*}
where $Z' = \frac{\partial}{\partial p} Z$ and $Z'' = \frac{\partial^2}{\partial p^2} Z$. Putting in~\eqref{eq:2nd-der-h} we obtain
\begin{align*}  \frac{\partial^2}{\partial p ^2}  H(\mathbf{F}) =& - \frac{Z'^2}{4Z(1-Z)}
-\frac{2Z'' Z - Z'^2}{8 Z^{3/2}} \left( \ln(1+ \sqrt{Z}) - \ln(1- \sqrt{Z})    \right).
\end{align*}
Using the Taylor expansions $\ln (1+Z) = \sum_{k=1}^{\infty}  \frac{(-1)^{k+1}}{k}Z^k$ and $\frac{1}{1-Z} = \sum_{k=0}^{\infty}  Z^k$ we find that
\begin{align*}  \frac{\partial^2}{\partial p ^2}  H(\mathbf{F})  =& \frac{Z''}{2}  + \left( \frac{Z''Z}{2}  - \frac{Z'^2}{4}    \right) \sum_{k=0}^{\infty} \frac{1}{2k+3}Z^k  + \frac{Z'^2}{4} \sum_{k=0}^{\infty} Z^k. \end{align*}
Finally using the definition of $Z$ we conclude that
\begin{align*}  
&\frac{\partial^2}{\partial p ^2}  H(\mathbf{F})  = -\|\overrightarrow{s} - \overrightarrow{r}\|^2\bigg[ \frac{\Delta}{\|\overrightarrow{s} - \overrightarrow{r}   \|^2} \sum_{k=0}^{\infty} \frac{1}{2k+3}Z^k + \left( 1- \frac{\Delta}{\|\overrightarrow{s}- \overrightarrow{r}\|^2}  \right) \sum_{k=0}^{\infty} Z^k  \bigg]. \end{align*}

A classical channel $X\rightarrow D$ for $|\mathcal{D}| = 2$ is determined by
$$p(D=0| X=0)=a, \quad p(D=0| X=1)=b.$$
We denote such a $D$ by $D_{a, b}$. Then $H(D_{a, b})  = h(ap+ b\bar{p})$ and
\begin{align*} &\frac{\partial^2}{\partial p^2}  H(D_{a, b})= -\frac{(a-b)^2}{ -(a-b)^2p^2 + (a-b)(1-2b)p + b(1-b) }\end{align*}
Let us assume that
\begin{align}\label{eq:condition-ab}
& a  \langle \overrightarrow{r} \vert \overrightarrow{s} - \overrightarrow{r}\rangle = \left( \|\overrightarrow{s} - \overrightarrow{r}\|^2 + \langle \overrightarrow{r} \vert \overrightarrow{s} - \overrightarrow{r}\rangle   \right) b - \frac{\|\overrightarrow{s} - \overrightarrow{r}\|^2}{ 2  }.
\end{align}
Then
\begin{align*}
\frac{\partial^2}{\partial p^2} H(D_{a,b})
 & = - \frac{\|\overrightarrow{s}- \overrightarrow{r}\|^2}{ \left(\frac{b(1-b) \|\overrightarrow{s}-\overrightarrow{r}\|^2}{(a-b)^2} + \|\overrightarrow{r}\|^2\right)  - Z       }\\
& = -\|\overrightarrow{s} - \overrightarrow{r}\|^2   \sum_{k=0}^{\infty} \frac{Z^k}{\left(  \frac{b(1-b) \|\overrightarrow{s}-\overrightarrow{r}\|^2}{(a-b)^2} + \|\overrightarrow{r}\|^2  \right)^{k+1}}.
\end{align*}
We now claim that for every $0\leq \theta\leq 1$ there exist $0\leq a_{\theta}, b_{\theta}\leq 1$ that satisfy~\eqref{eq:condition-ab} and
$$ \frac{b_{\theta}(1-b_{\theta}) \|\overrightarrow{s}-\overrightarrow{r}\|^2}{(a_{\theta}-b_{\theta})^2} + \|\overrightarrow{r}\|^2   = \frac{1}{\theta}.$$
By continuity we only need to prove the claim for $\theta=0$ and $\theta=1$. For $\theta=0$ take $a_0=b_0=1/2$, and
for $\theta=1$ take $a_1$ to be equal to
\[ \frac{1}{2}+   \frac{1}{2}\frac{\|\overrightarrow{s} - \overrightarrow{r}\|^2 +  \langle \overrightarrow{r}\vert \overrightarrow{s} - \overrightarrow{r}\rangle }{ \sqrt{\| \overrightarrow{s} - \overrightarrow{r} \|^2\cdot (1- \|\overrightarrow{r}\|^2) + \langle \overrightarrow{r} \vert \overrightarrow{s}- \overrightarrow{r}\rangle^2}} ,
\]
and $b_1$ to be equal to
\[ \frac{1}{2}+\frac{1}{2}\frac{\langle \overrightarrow{r}\vert \overrightarrow{s} - \overrightarrow{r}\rangle}{\sqrt{ \|\overrightarrow{s} - \overrightarrow{r}\|^2 \cdot(1 - \|\overrightarrow{r}\|^2) + \langle \overrightarrow{r} \vert \overrightarrow{s} - \overrightarrow{r}\rangle^2  }}.
\]
Using $\|\overrightarrow{r}\|, \|\overrightarrow{s}\| \leq 1$, it is easy to see that $0\leq a_1, b_1\leq 1$.
We thus have
\[
\frac{\partial^2}{\partial p^2}  H(D_{a_{\theta}, b_{\theta}}) = -\|\overrightarrow{s} - \overrightarrow{r}\|^2 \sum_{k=0}^{\infty} \theta^{k+1} Z^k
\]

Now define a channel $X\rightarrow C$ which with probability $1-\frac{\Delta}{\| \overrightarrow{s}-\overrightarrow{r} \|^2}$ is equal to $X\rightarrow D_{a_1, b_1}$, and with probability $\frac{\Delta}{\| \overrightarrow{s}-\overrightarrow{r} \|^2}$ is equal to the channel $X\rightarrow D_{a_{\theta^2}, b_{\theta^2}}$ where $0\leq \theta\leq 1$ is chosen uniformly at random. Observe that
\begin{align*}
\frac{\partial^2}{\partial p^2} H(C)  & = -\frac{\Delta}{\| \overrightarrow{s}-\overrightarrow{r} \|^2} \int_0^1 \frac{\partial^2}{\partial p^2} H(D_{a_{\theta^2}, b_{\theta^2}}) \mathrm{d}\theta  
- \left(1- \frac{\Delta}{\| \overrightarrow{s}-\overrightarrow{r} \|^2}  \right) \frac{\partial^2}{\partial p^2} H(D_{a_1, b_1})\\
& = -\|\overrightarrow{s} - \overrightarrow{r}\|^2\left[   \frac{\Delta}{\| \overrightarrow{s}-\overrightarrow{r} \|^2} \int_0^1 \sum_{k=0}^\infty \theta^{2k+2}Z^k \mathrm{d}\theta  
+ \left(1- \frac{\Delta}{\| \overrightarrow{s}-\overrightarrow{r} \|^2}  \right) \sum_{k=0}^{\infty} Z^k    \right]\\
& = -\|\overrightarrow{s} - \overrightarrow{r}\|^2\left[   \frac{\Delta}{\| \overrightarrow{s}-\overrightarrow{r} \|^2}  \sum_{k=0}^\infty \left(\int_0^1 \theta^{2k+2} \mathrm{d}\theta \right) Z^k  
+  \left(1- \frac{\Delta}{\| \overrightarrow{s}-\overrightarrow{r} \|^2}  \right) \sum_{k=0}^{\infty} Z^k    \right]\\
& = -\|\overrightarrow{s} - \overrightarrow{r}\|^2\left[   \frac{\Delta}{\| \overrightarrow{s}-\overrightarrow{r} \|^2}  \sum_{k=0}^\infty \frac{1}{2k+3} Z^k  
+ \left(1- \frac{\Delta}{\| \overrightarrow{s}-\overrightarrow{r} \|^2}  \right) \sum_{k=0}^{\infty} Z^k    \right] \\
& = \frac{\partial^2}{\partial p^2}  H(\mathbf{F}).
\end{align*}

\end{document}